\documentclass[12pt]{article}
\usepackage[margin=1in]{geometry} 
\usepackage{amsfonts}
\usepackage{mathrsfs}
\usepackage{amssymb}
\usepackage{booktabs}
\usepackage{tabularx}
\usepackage{hyperref}
\hypersetup{hidelinks}
\usepackage{caption}
\usepackage[linesnumbered,ruled,vlined]{algorithm2e}
\usepackage{natbib}
\usepackage{bm}
\usepackage{amsmath}
\usepackage{amsthm}
\usepackage{multirow}
\usepackage{cases}
\usepackage{graphicx}
\usepackage{color}
\usepackage{rotating}
\usepackage{tabularray}
\usepackage{subcaption}
\usepackage{ragged2e}  
\usepackage{float}
\newcommand{\argmin}{\operatorname*{argmin}}
\UseTblrLibrary{booktabs}
\newtheorem{thm}{\textbf{Theorem}}
\newtheorem{lem}{\textbf{Lemma}}
\newtheorem{Def}{\textbf{Definition}}
\newtheorem{remark}{\textbf{Remark}}
\newtheorem{exam}{\textbf{Example}}
\usepackage{caption}  
\begin{document}

\title{\bf\Large  Parameter estimation and application in two types of uncertain single-index models}

\author{\normalsize Fuguo Wang, Zhiming Li\footnote{Corresponding author: zmli@xju.edu.cn}\\
{\em { \small College of Mathematics and System Science, Xinjiang University,  Urumqi 830046, China}}}
\date{}
\maketitle

\begin{flushleft}
{\bf Abstract} Uncertain data often arises in complex environments because of frequency instability and subjective judgment. This paper establishes two types of uncertain single-index models to capture the inherent properties of such data. Based on the semiparametric least-squares principle,  the Nadaraya–Watson kernel and B-spline methods are used to estimate the unknown coefficients in various scenarios with both crisp and imprecise explanatory variables. Residual analysis and hypothesis testing under uncertainty assess the fit of the proposed models. Furthermore, simulation studies verify the models' validity, and a real-data application demonstrates their effectiveness in practical settings.
\end{flushleft}
\begin{flushleft}
\end{flushleft}

\noindent \hrulefill

\section{Introduction}
Single-index models are an important class of semiparametric models that mitigate the curse of dimensionality in nonparametric estimation by projecting multivariate covariates onto a univariate index, while retaining the interpretability of parametric models. Due to their strong performance in modeling high-dimensional, complex data, single-index models have been widely applied in biomedicine (Zhang et al., \citeyear{Zhang2023}), the environment (Zhang et al., \citeyear{Zhang2025}), and finance (Zhu et al., \citeyear{Zhu2026}). To date, outstanding theoretical results have been achieved for single-index models, especially in coefficient estimation and testing.  In terms of estimation, Stoker (\citeyear{Stoker1986}) pioneered the average derivative estimator.  Ichimura (\citeyear{Ichimura1987}) extended the least squares method to a semiparametric framework, and further developed the weighted semiparametric least squares estimator to address heteroscedasticity (Ichimura, \citeyear{Ichimura1993}). Since then, numerous alternative estimation approaches have been successively developed in the literature, including the semiparametric maximum likelihood estimator (Ai, \citeyear{Ai1997}), the partial least squares estimator (Naik and Tsai,  \citeyear{NaikTsai2000}), Bayesian estimation (Antoniadis et al., \citeyear{Antoniadis2004}), spline estimation (Wang and Yang, \citeyear{WangYang2009}), the estimating functions method (EFM) (Cui et al., \citeyear{Cui2011}),  and  Lasso estimator (Zhu et al., \citeyear{Zhu2026}). With the enrichment of estimation methods, the development of testing theories has greatly contributed to the maturity of the statistical inference procedure for single-index models, such as the goodness-of-fit test (Xia et al., \citeyear{Xia2004test}; Stute and Zhu, \citeyear{StuteZhu2005}),  two-sample test (Zhang et al., \citeyear{zhang2018constructive}), coefficient significance test (Cai et al., \citeyear{cai2023test}), joint hypothesis test for model parameters with both parametric and nonparametric components (Zhou et al., \citeyear{Zhou2025}).

 Within the probabilistic framework, single-index models are well-suited and relatively flexible for crisp data of both numerical and continuous types. However, imprecision is often observed in practical settings due to frequency instability and expert ratings/judgments. Single-index models, based on probability measures, may lead to information loss, estimation bias, and erroneous statistical inferences when handling imprecise data. To tackle the problem of uncertain data, Liu (\citeyear{Liu2007}) founded the uncertainty theory based on uncertainty measure with normality, duality, subadditivity, and product axioms, which was further refined by Liu (\citeyear{Liu2010}). The novel measure can be used to evaluate imprecise interval observations by accounting for uncertain variables and their distributions. Thus, it serves as a powerful tool for the classification (Liang and Liu, \citeyear{Liang2025}) and prediction tasks with imprecise data. Recent progress in uncertainty theory is discussed by Liu (\citeyear{Liu2026}).  As a vital branch of uncertain statistics in uncertainty theory, uncertain regression analysis was initiated by Yao and Liu (\citeyear{YaoLiu2018}) and serves to characterize the relationships among variables within uncertain environments. Lio and Liu (\citeyear{LioLiu2018}) established the method of moments to determine the disturbance term. To alleviate the sensitivity of estimation to outliers, Liu and Yang (\citeyear{LiuYang2020}) put forward the least absolute deviation estimator.  Lio and Liu (\citeyear{LioLiu2020})  proposed the uncertain maximum likelihood estimation, which was extended by Liu and Liu (\citeyear{LiuLiu2024}) to improve robustness against outliers. Moreover, advances in uncertainty testing techniques have further consolidated the theoretical framework of uncertain statistics. In the context of hypothesis testing, Ye and Liu (\citeyear{YeLiu2022}) developed uncertain hypothesis tests and applied them to test the disturbance term and assess model adequacy in uncertain regression. Subsequently,  Ye and Liu (\citeyear{YeLiu2023}) proposed uncertain significance tests to examine whether regression coefficients are statistically zero. We note that early studies mainly focused on parametric regression models, forming a relatively complete parametric modeling system.  With advances in research, nonparametric uncertain regression has been gradually developed. For instance, Ding and Zhang (\citeyear{DingZhang2021}) adopted B-spline and local polynomial smoothing methods to approximate unknown nonlinear functions, providing a feasible solution for nonparametric modeling in uncertain scenarios. Building on previous studies, Zhang and Li (\citeyear{ZhangLi2026}) extended to semiparametric models and established a statistical inference procedure for uncertain semi-varying coefficient models. Although many research frameworks cover both parametric and nonparametric models with uncertainty, systematic research on single-index models within the uncertainty theory remains limited, despite their unique advantages in dimensionality reduction and interpretability. 
 
 Based on the above analysis, we will propose uncertain single-index models within the framework of uncertainty theory and further conduct the estimation methods and testing procedures. The main contributions are summarized as follows:
 
  (i) Two types of uncertain single-index models are established to achieve greater flexibility in modeling nonlinear relationships and provide a novel and powerful tool for the nonlinear modeling of high-dimensional imprecise data.
  
  (ii) Parameter estimation methods are proposed for different types of single-index models. Particularly for cases where explanatory variables take imprecise data, the monotonicity constraints to be satisfied by the approximate function are specified, and the explicit form of the inverse distribution of explanatory variables is derived. The hypothesis testing procedures are conducted for residual analysis. 
  
  (iii) Numerical simulations are presented to illustrate the performance of single-index models. As an application, Lagos weather datasets are used to verify the effectiveness of the proposed model and methods.

  The rest of the paper is organized as follows. Section \ref{sec2} first reviews the basic concepts of uncertain theory and then introduces the definitions of uncertain single-index models. In Section \ref{sec3}, we establish uncertain single-index models with crisp and uncertain explanatory variables and conduct the estimation method and residual analysis. Numerical simulation in Section \ref{sec5} shows the performance of the proposed methods in the two models above-mentioned. The proposed model is applied to a weather dataset in Section \ref{sec6}. Finally, a brief conclusion is given in Section \ref{sec7}.

\section{Uncertain single-index models}\label{sec2}

In this section, we first review basic concepts of uncertainty theory proposed by Liu (\citeyear{Liu2026}).  Let $\Gamma$ be a nonempty set, and $\mathcal{L}$ a $\sigma$-algebra over $\Gamma$. Each element $\Lambda$ in $\mathcal{L}$ is called an event. A set map $\mathcal{M}$: $\mathcal{L}\rightarrow[0,1]$ is an uncertain measure, following four axioms:
	
{\bf Axiom 1} (Normality Axiom) $\mathcal{M}\{\Gamma\}=1$ for the universal set $\Gamma$.
	
{\bf Axiom 2} (Duality Axiom) $\mathcal{M}\{\Lambda\}+\mathcal{M}\{\Lambda^c\}=1$ for any event
$\Lambda$.

{\bf Axiom 3} (Subadditivity Axiom) For every countable sequence of events $\Lambda_1,\Lambda_2,...,$
$$
\mathcal{M}\Big\{\bigcup_{i=1}^\infty\Lambda_i\Big\}\leq\sum_{i=1}^\infty\mathcal{M}\{\Lambda_i\}.
$$

{\bf Axiom 4} (Product Axiom) Let $(\Gamma_i,\mathcal{L}_i,\mathcal{M}_i)$ be uncertainty spaces for $i=1, 2,...,$ the product uncertain measure $\mathcal{M}$ is an uncertain measure satisfying
$$\mathcal{M}\Big\{\prod_{i=1}^\infty\Lambda_i\Big\}=\bigwedge_{i=1}^\infty\mathcal{M}_i\{\Lambda_i\}.$$
 
An uncertain variable is a function $\xi$ from an uncertainty space $(\Gamma, \mathcal{L}, \mathcal{M})$ to the set of real numbers $\mathbb{R}$, such that for any Borel set B $\in$ $\mathbb{R}$, the set $\{\xi \in B\} = \{\gamma \in \Gamma \mid \xi(\gamma) \in B\}$ is an event in the uncertainty space. The uncertainty distribution $\Phi: \mathbb{R} \to [0, 1]$ of $\xi$ is defined by $\Phi(x) = \mathcal{M}\{\xi \leq x\}$ for any real number $x \in \mathbb{R}$.  An uncertainty distribution $\Phi(x)$ is said to be regular if it is a continuous and strictly increasing function with respect to $x$ at which $0 < \Phi(x) < 1$, and $\lim_{x \to -\infty} \Phi(x) = 0, \lim_{x \to +\infty} \Phi(x) = 1.$
 If the uncertainty distribution $\Phi(x)$ is regular, then the inverse function $\Phi^{-1}(\alpha)$ is called the inverse uncertainty distribution of $\xi$.
The expectation and variance of the variable $\xi$ are defined as
$$E[\xi] = \int_0^1 \Phi^{-1}(\alpha) \, \mathrm{d}\alpha, \quad V[\xi] = \int_0^1 (\Phi^{-1}(\alpha)-E[\xi])^2 \, \mathrm{d}\alpha.$$

An uncertain variable $\xi$ is called linear if it has a linear uncertainty distribution
\[
\Phi(x) = 
\begin{cases}
0, & \text{if } x \leq a, \\[6pt]
\dfrac{x - a}{b - a}, & \text{if } a < x \leq b, \\[6pt]
1, & \text{if } b < x,
\end{cases}
\]
denoted by $\mathcal{L}(a,b)$, where $a$ and $b$ are real numbers with $a < b$.
The inverse uncertainty distribution is $\Phi^{-1}(\alpha) = (1 - \alpha)a + \alpha b.$

An uncertain variable $\xi$ is called normal if it has a normal uncertainty distribution
\[
\Phi(x) = \left(1 + \exp\left(\frac{\pi(e - x)}{\sqrt{3}\sigma}\right)\right)^{-1}, \quad x \in \mathbb{R}, 
\]
denoted by $\mathcal{N}(e, \sigma)$, where $e$ and $\sigma$ are real numbers with $\sigma > 0$. A normal uncertainty distribution is called standard if $e = 0$ and $\sigma = 1$. The inverse uncertainty distribution of normal uncertain variable $\mathcal{N}(e, \sigma)$ is
\[
\Phi^{-1}(\alpha) = e + \frac{\sqrt{3}\sigma}{\pi} \ln \frac{\alpha}{1 - \alpha}, \quad \alpha \in (0,1). 
\]

\begin{lem}
\label{lem21}{\rm(Liu, \citeyear{Liu2026})}
{\rm
Let $\xi_1, \dots, \xi_n$ be uncertain variables, and $f$ be a real-valued measurable function. Then, $f(\xi_1, \dots, \xi_n)$ is an uncertain variable.
}
\end{lem}

\begin{lem}
\label{lem22}{\rm(Liu, \citeyear{Liu2026})}
{\rm
Let $\xi_1, \dots, \xi_n$ be independent uncertain variables with regular uncertainty distributions $\Phi_1, \dots, \Phi_n$, respectively. If the function $f$ is strictly increasing with respect to $\xi_1, \dots, \xi_m$ and strictly decreasing with respect to $\xi_{m+1}, \xi_{m+2}, \dots, \xi_n$, then the inverse uncertainty distribution of the uncertain variable $\xi = f(\xi_1,\dots, \xi_n)$ is
\[
\Psi^{-1}(\alpha) = f\bigl(\Phi_1^{-1}(\alpha), \dots, \Phi_m^{-1}(\alpha), \Phi_{m+1}^{-1}(1-\alpha), \dots, \Phi_n^{-1}(1-\alpha)\bigr).
\]
}
\end{lem}

\begin{lem}
\label{lem23}{\rm(Liu, \citeyear{Liu2026})}
{\rm
Let $\xi$ be an uncertain variable with regular uncertainty distribution $\Phi$. Then, the $k$th moment of the variable $\xi$ can be calculated as
$E[\xi^k] = \int_0^1 \left( \Phi^{-1}(\alpha) \right)^k \, \mathrm{d}\alpha.$
}
\end{lem}

For uncertain response observations, we employ uncertainty theory to characterize the disturbance term, thereby establishing two types of uncertain single-index models with a crisp explanatory vector and an uncertain one. Let $\boldsymbol{\beta}=(\beta_1,\ldots,\beta_p)^T$, and assume $\|\boldsymbol{\beta}\|$=$\sqrt{\beta_1^2 + \cdots + \beta_p^2}=1$ with
$\beta_1>0$.
\begin{Def}
{\rm
Suppose that $g$ is an unknown univariate function. Let $\tilde{y}$ be an uncertain response variable, and $\epsilon$ be an uncertain disturbance term. 

(i) If $\boldsymbol{{x}}=(x_1,\ldots,x_p)$ is a \textbf{crisp} explanatory vector,  the relationship between $\tilde{y}$ and  $\boldsymbol{{x}}$ is expressed by
\begin{equation}
\tilde{y}=g(\boldsymbol{\beta}^T\boldsymbol{{x}})+\epsilon=g(\sum_{k=1}^{p}\beta_kx_{k})+\epsilon,
\label{eq:1}
\end{equation}
called an uncertain single-index model with \textbf{crisp} explanatory vector (USIC model). 

(ii)  If $\boldsymbol{\tilde{x}}=(\tilde{x}_1,\ldots,\tilde{x}_p)$ is a \textbf{uncertain} explanatory vector, we call the following form
\begin{equation}
\tilde{y}=g(\boldsymbol{\beta}^T\boldsymbol{\tilde{x}})+\epsilon=g(\sum_{k=1}^{p}\beta_k\tilde{x}_{k})+\epsilon,
\label{eq:2}
\end{equation}                                           
an uncertain single-index model with \textbf{uncertain} explanatory vector (USIU model).                                                
}
\end{Def}

In the two uncertain models (\ref{eq:1}) and (\ref{eq:2}), we first estimate the unknown function $g$ and index coefficients $\boldsymbol{\beta}$, and then conduct the residual analysis of the uncertain variable $\tilde{y}$ in the following sections. 

\section{Profile least squares estimation}\label{sec3}

In this section, we propose two methods: the uncertain Nadaraya-Watson(N-W) kernel method and the B-spline method, under the profile least squares framework, to estimate the unknown function $g$ and index coefficient $\boldsymbol{\beta}$ in the USIC and USIU models. 

\subsection{Uncertain N-W kernel estimation in the USIC model}
\label{Sect:3.1}
For the USIC model (\ref{eq:1}), let $(\boldsymbol{{x}_i},\tilde{y}_i)$ $(i=1,2,...,n)$ be the observed data, where $\boldsymbol{{x}_i}=(x_{i1}, x_{i2}, \ldots, x_{ip})^T(i=1,\ldots, n)$ is the $i$th $p$-dimensional explanatory vector,  and $\tilde{y}_i$ is the uncertain response variable.  Then, based on the model (\ref{eq:1}), we have
	\begin{equation}
		\tilde{y}_i=g(\boldsymbol{\beta}^T\boldsymbol{{x}_i})+\epsilon_i=g(\sum_{k=1}^{p}\beta_kx_{ik})+\epsilon_i, i=1,\ldots, n,
		\label{eq:3}
\end{equation}
Let $\hat g$ and $\hat{\bm{\beta}}$ be the estimators of $g$ and $\bm{\beta}$ in the model (\ref{eq:3}).  By  profile least-squares method,  $\hat g$ and $\hat{\bm{\beta}}$ are the solutions to the minimization problem
\begin{equation}
(\hat g, \hat{\bm{\beta}})=\argmin_{g,\bm{\beta}} \sum_{i=1}^{n} E\bigl[(\tilde{y}_i - g(\sum_{k=1}^{p}\beta_kx_{ik}))^2\bigr].
\label{eq:4}
\end{equation}

	

\begin{thm}
{\rm	
In the USIC model (\ref{eq:3}), the minimization problem (\ref{eq:4}) is equivalent to the optimization problem
\begin{equation}
(\hat g, \hat{\bm{\beta}})=\min_{g,\bm{\beta}} \sum_{i=1}^{n} (E[\tilde{y}_i]-g(\sum_{k=1}^{p}\beta_kx_{ik}))^2.
\label{eq:5}
\end{equation}
}
\end{thm}

\begin{proof}
Following the stipulation (2) of \citet{YaoLiu2018}, we have
\begin{equation*}
 E\bigl[(\tilde{y}_i-g(\sum_{k=1}^{p}\beta_kx_{ik}))^2\bigr]=V[\tilde{y_i}]+(E[\tilde{y_i}]-g(\sum_{k=1}^{p}\beta_kx_{ik}))^2.
\end{equation*}
Furthermore, we get the following equation
\begin{equation*}
\sum_{i=1}^{n}E\bigl[(\tilde{y}_i-g(\sum_{k=1}^{p}\beta_kx_{ik}))^2\bigr]=\sum_{i=1}^{n}V[\tilde{y_i}]+\sum_{i=1}^{n}(E[\tilde{y_i}]-g(\sum_{k=1}^{p}\beta_kx_{ik}))^2.
\end{equation*}
Given observed data, the sum of $V[\tilde{y}_i]$ is fixed. Thus, the minimization problem~(\ref{eq:4}) is equivalent to the optimization (\ref{eq:5}).
\end{proof}
To estimate the unknown function $g$ and index coefficients $\bm{\beta}$ in (\ref{eq:4}), we adopt a two-step iterative procedure as follows: \\
\indent \textbf{Step 1}: Estimate the index coefficients using a general link function (e.g., the identity function). \\
\indent \textbf{Step 2}: Treat the direction in the predictor space determined by the index coeffcient as known and estimate the unknown function $g$ via a nonparametric smoother. \\
\indent \textbf{Step 3}: Reestimate the index coefficients based on the estimated link function, and iterate between the two steps.

Before initiating the iterative procedure, a suitable approximation method of $g$ must be obtained. Several approaches are available, and we adopt the uncertain Nadaraya-Watson (N-W) kernel method for its computational simplicity to address the least-squares problem (\ref{eq:4}). For a fixed ${\bm{\beta}}$, the uncertain N-W kernel estimator of unknown function $g$ evaluated at $\bm{{\beta}^T{x}}$ can be expressed as
\begin{equation}
\hat{g}(\bm{{\beta}}^T\bm{{x}}|\bm{\beta})=\frac{\sum_{j=1}^{n} E[\tilde{y}_j] K_h\left(\bm{\beta}^T \bm{x} - \bm{\beta}^T\bm{ x_j}\right)} {\sum_{j=1}^{n} K_h\left(\bm{\beta}^T \bm{x} - \bm{\beta}^T \bm{x_j}\right)}.
\label{eq:6}
\end{equation}
In the estimation of $\bm{\beta}$, the leave-one-out uncertain N-W estimator is commonly adopted to eliminate the bias caused by using the $i$th observation itself in kernel smoothing:
\begin{equation}
\hat{g}^{(-i)}(\bm{{\beta}}^T\bm{{x_i}}|\bm{\beta})=\frac{\sum_{j \neq i} E[\tilde{y}_j] K_h\left(\bm{\beta}^T \bm{x_i} - \bm{\beta}^T \bm{x_j}\right)}{\sum_{j \neq i} K_h\left(\bm{\beta}^T \bm{x_i} - \bm{\beta}^T \bm{x_j}\right)},
\label{eq:7}
\end{equation}
where $K_h(u)=h^{-1}K(u/h)$ denotes the scaled kernel function, with $K(\cdot)$ being a fixed base kernel, and $h>0$ is the bandwidth controlling the range of local averaging. 
Substituting the unknown function $g$ with the leave-one-out N-W estimator (\ref{eq:3}), the problem (\ref{eq:4}) reduces to
\begin{equation}
\hat{\bm{\beta}}=\argmin_{\bm{\beta}} \sum_{i} (E[\tilde{y}_i]-	\hat{g}^{(-i)}(\bm{{\beta}}^T\bm{{x_i}}|\bm{\beta}))^2.
\label{eq:8}
\end{equation}

\begin{remark}
{\rm
To avoid the problem that the denominators may be close to zero in equation (\ref{eq:7}), let $A \subseteq \mathbb{R}^p$ be a chosen set, and $x_A$ be the elements of the set $A$, we assume $x_A$ has the same distribution as $x$ when $x \in A$. In equation (\ref{eq:8}), $\sum_{i}$ represents the sum taken over all indices $i$ for which $x_i \in A$.
}
\end{remark}

The kernel function performs a weighted average of the single indices $\{ \bm{\beta}^T \bm{x_j} \}_{j \neq i}$ over data points $\bm{x_j}$ at different positions, assigning larger weights to data points for which the single-index $\bm{\beta}^T \bm{x_j}$ is closer to $\bm{\beta}^T\bm{x_i}$. For this purpose, the kernel function is usually chosen to satisfy the conditions of non-negativity, normalization, and symmetry. Within the uncertainty theory framework, the derivative function of the standard uncertain normal distribution is a suitable choice due to its favorable theoretical properties and good empirical performance. The corresponding mathematical form is:
\begin{equation}
K(u) = \frac{\pi}{\sqrt{3}} \exp\left(-\frac{\pi u}{\sqrt{3}}\right) \left(1 + \exp\left(-\frac{\pi u}{\sqrt{3}}\right)\right)^{-2}.
\label{eq:9}
\end{equation}

In kernel regression analysis, the performance of the N-W kernel estimation method depends heavily on the selection of the bandwidth $h$. As a critical parameter of the kernel function, bandwidth directly determines the trade-off between bias and variance of the estimation results. An excessively large $h$ causes kernel over-smoothing, leading the estimator to overlook local characteristics and thus incur substantial bias. Conversely, an overly small $h$ makes the estimator over-sensitive to data noise and outliers, leading to overfitting and a sharp increase in variance. For selecting model parameters, V-fold cross-validation is recommended for its effectiveness, as discussed by Ding and Zhang (\citeyear{DingZhang2021}). This method is mathematically expressed as:
\begin{equation}
 {CV}(h) = \frac{1}{V} \sum_{v=1}^{V} \sum_{i \in I_{-v}} E( \tilde{y}_i - \hat{g}^{(-i)}(\bm{{\beta}}^T\bm{{x_i}}|\bm{\beta}))^2,
\label{eq:10}
\end{equation}
where  $\hat{g}^{(-i)}(\bm{{\beta}}^T\bm{{x_i}}|\bm{\beta})$ is obtained from the $v$th training subset, and $I_{-v}$ stands for the corresponding validation subset. The optimal parameter $h$ is derived by minimizing $CV(h)$, i.e., $\hat{h} = \arg\min_{h} {CV}(h)$.

Owing to the complexity of Equation (\ref{eq:10}), an analytical solution for its minimizer cannot be easily derived. We employ the Fibonacci search algorithm as the core optimization tool for its high efficiency and precision, and it is well-suited for the unimodal cross-validation function minimization task corresponding to Equation (\ref{eq:10}). The detailed implementation of the algorithm is presented in Algorithm \ref{alg:fibonacci_bandwidth}.
\begin{algorithm}[t]
	\caption{Calculate the optimal bandwidth $h^*$ }
	\label{alg:fibonacci_bandwidth}
	\SetAlgoLined 
	\SetKwInput{KwInput}{Input}
	\SetKwInput{KwOutput}{Output}
	\KwIn{
		The dataset $D = (\tilde{x}_1, \cdots, \tilde{x}_p, \tilde{y})$, 
		the search interval $[h_{\text{min}}, h_{\text{max}}]$,\\
		\quad \quad \quad \quad
		the maximum number of iterations $n$, 
		the precision threshold $\epsilon$
	}
	\KwOut{The optimal bandwidth $h^*$ and its cross-validation value $CV(h^*)$}
	
	$\text{fib}(0) \gets 0$, $\text{fib}(1) \gets 1$\;
	\For{$k = 2$ \KwTo $n$}{
		$\text{fib}(k) \gets \text{fib}(k-1) + \text{fib}(k-2)$\;
	}

	$\rho \gets 1 - \frac{\text{fib}(n-1)}{\text{fib}(n)}$\;
	$h_l \gets h_{\text{min}} + \rho \cdot (h_{\text{max}} - h_{\text{min}})$,\quad	$h_r \gets h_{\text{max}} - \rho \cdot (h_{\text{max}} - h_{\text{min}})$\;
	
	Calculate $CV(h_l)$ and $CV(h_r)$ from Equation (\ref{eq:9})\;
	$\text{tol} \gets |CV(h_{\text{max}}) - CV(h_{\text{min}})|$\;
	
	\While{$n > 2$ \textbf{and} $\text{tol} > \epsilon$}{
		$n \gets n - 1$,\quad$\rho \gets 1 - \frac{\text{fib}(n-1)}{\text{fib}(n)}$\;
		
		\eIf{$CV(h_l) < CV(h_r)$}{
			$h_{\text{max}} \gets h_r$\;
			$h_r \gets h_l$,\quad$CV(h_r) \gets CV(h_l)$\;
			$h_l \gets h_{\text{min}} + \rho \cdot (h_{\text{max}} - h_{\text{min}})$\;
			Calculate $CV(h_l)$\;
		}
		{
			$h_{\text{min}} \gets h_l$\;
			$h_l \gets h_r$,\quad$CV(h_l) \gets CV(h_r)$\;
			$h_r \gets h_{\text{max}} - \rho \cdot (h_{\text{max}} - h_{\text{min}})$\;
			Calculate $CV(h_r)$\;
		}
		$\text{tol} \gets |CV(h_{\text{max}}) - CV(h_{\text{min}})|$\;
	}
	
	$h^* \gets \frac{h_{\text{min}} + h_{\text{max}}}{2}$,\quad Calculate $CV(h^*)$\;
	\Return Optimal bandwidth $h^*$ and $CV(h^*)$\;
\end{algorithm}

Based on the link function and coefficients estimated by Eq.(\ref{eq:7}) and (\ref{eq:8}), we define the $i$th residual 
\begin{equation}
\hat{\epsilon}_i=\tilde{y}_i -\hat{g}^{(-i)}(\sum_{k=1}^{p}\hat{\beta}_kx_{ik}),\quad i = 1, \dots, n.
\label{eq:11}
\end{equation}
Under the assumptions that the disturbance terms $\epsilon_1,\ldots,\epsilon_n$ are independent and that their expectation and variance exist, the expectation and variance of $\hat{\epsilon}$ can be estimated as:
\begin{equation}
\hat{e} = \frac{1}{n} \sum_{i=1}^{n} E\left[\hat{\epsilon}_i\right],\quad 
 \hat{\sigma}^2 = \frac{1}{n} \sum_{i=1}^{n} E(\hat{\epsilon}_i - \hat{e})^2.
 \label{eq:12}
\end{equation}

Given the new $\bm{x}=(x_1, x_2,\ldots, x_p)^T$,  substituting leave-one-out N-W estimator with its general form $(\ref{eq:6})$ and following Liu (\citeyear{Liu2026}), the forecast uncertain variable of $\tilde{y}$ can be expressed as
$
\hat{y}=\hat{g}(\hat{\bm{\beta}}^T\bm{x})+\hat{\epsilon}.
$
Thus, the forecast value $\hat{\mu}$ is obtained by
\begin{equation}
\hat{\mu}=E[\hat{y}]=\hat{g}(\hat{\bm{\beta}}^T\bm{x})+\hat{e}.
\label{eq:13}
\end{equation}
Assume that $\hat{y}$ and the disturbance term $\epsilon$ follow the uncertainty distributions $\hat{\Psi}$ and $\Omega$, respectively. The inverse uncertainty distribution of  $\hat{y}$ is
$
\hat{\Psi}^{-1}(\alpha) =\hat{g}(\hat{\rm\beta}^T\bm{x}) + \Omega^{-1}(\alpha).
$
Given a confidence level $\alpha$,   let $a^{*}$ be the minimum value such that
\begin{equation}
\hat{\Psi}(\hat{\mu} + a^{*}) - \hat{\Psi}(\hat{\mu} - a^{*}) \geq \alpha.
\label{eq:14}
\end{equation}
It leads that the confidence interval for response variable $\hat{y}$ is $[\hat{\mu}-a^{*}, \hat{\mu}+a^{*}]$.

\subsection{B-spline estimation in the USIU model}
\label{Sect:3.2}

When both of the explanatory and response variables are uncertain, solving the least-squares problem requires imposing a strict monotonicity constraint on the approximation function. Since the N-W estimator is uncertain, it lacks an explicit functional structure, making it hard to enforce a rigorous global monotonicity constraint. Therefore, we adopt the polynomial spline method, in particular the B-spline method, for its explicit functional form and for the ease of differentiating the approximation function. 

In the USIU model (\ref{eq:2}), let $(\boldsymbol{\tilde{x}_i},\tilde{y}_i)$ $(i=1,\ldots,n)$ be a set of observed data, where $\boldsymbol{\tilde{x}_i}=(\tilde{x}_{i1},\ldots,\tilde{x}_{ip})^T$ denotes independent uncertain explanatory vector with regular uncertainty distribution $\Phi_{ik}$, and $\tilde{y}_i$ are independent response variables with regular uncertainty distribution $\Psi_{i}$ for $i=1,\ldots,n$ and $k=1,\ldots,p$. From the model ({\ref{eq:2}}), it follows that
	\begin{equation}
		\tilde{y}_i=g(\boldsymbol{\beta}^T\boldsymbol{\tilde{x}_i})+\epsilon_i=g(\sum_{k=1}^{p}\beta_k\tilde{x}_{ik})+\epsilon_i,i=1,\ldots, n.
		\label{eq:15}
\end{equation}

Let $l$ be a positive integer for the B-spline degree, and \(c=t_1<t_2<\dots<t_q=d\) be a knot sequence that partitions the interval $[c,d]$, $m=q+l-1$ denotes the number of B-spline basis functions. Denote the basis functions vector $\bm{\delta}^T(\tilde{t})=(\delta_1(\tilde{t}),\ldots,\delta_m(\tilde{t}))^T$ and  ${\bm{b}}=({b}_1,\ldots,{b}_m)^T$. Then, the B-spline estimate of a smooth function $g$ in the USIU model (\ref{eq:15}) is
$
\hat{g}(\tilde{t}_i) = \bm{\delta}^T(\tilde{t}_i){\bm{b}},
$
where $\tilde{t}_i=\bm{\beta}^T \boldsymbol{\tilde{x}_i}$. Let $\hat{\bm{b}}=(\hat{b}_1,\ldots,\hat{b}_m)^T$ be the estimator of ${\bm{b}}$.  
Under these conditions, the least squares estimate of the unknown function $g$ and the parameters in model (\ref{eq:15}) solves
\begin{equation}
 (\hat{\bm{b}}, \hat{\bm{\beta}})=\argmin_{\bm{b},\bm{\beta}} \sum_{i=1}^n E[(\tilde{y}_i - \bm{\delta}^T(\bm{\beta}^T \boldsymbol{\tilde{x}_i})\bm{b})^2].
\label{eq:16}
\end{equation}

\begin{thm}
\label{thm:2}
{\rm 
Assume that the function $g$ in model (\ref{eq:15}) is of class $C^1$ and strictly increasing. The minimization problem (\ref{eq:16}) is reduced to the following problem
\begin{equation}
 (\hat{\bm{b}}, \hat{\bm{\beta}})=\argmin_{\bm{b}, \bm{\beta}}  \sum_{i=1}^{n} \int_0^1 ( \Psi_{i}^{-1}(\alpha) - \bm{\delta}^\mathrm{T}(\sum_{k=1}^{p}\beta_k\varphi_{ik}^{-1}(\alpha)\big)\bm{b} )^2 d\alpha, 
\label{eq:17}
\end{equation}
satisfying $( \bm{\delta}^\mathrm{T}(t) \bm{\hat{b}})' > 0 \text{ for } t \in \mathbb{R}$  when $l \geq 2$, where
\begin{equation}\label{varphi}
\varphi_{ik}(\alpha)=
\begin{cases} 
\Phi_{ik}(\alpha), & \beta_k < 0, \\
\Phi_{ik}(1-\alpha),  & \beta_k > 0,
\end{cases}
\quad i=1, \dots, n, \quad k=1,\dots, p.
\end{equation}
}
\end{thm}

\begin{proof} 
Let $f(\tilde{y}_i,\tilde{x}_{i1},\ldots,\tilde{x}_{ip})
=\tilde{y}_i-g(\sum_{k=1}^{p}\beta_k\tilde{x}_{ik})$. For each $i$, we have
\begin{equation*}
\dfrac{\partial f}{\partial \tilde{y}_i} = 1,\quad
\dfrac{\partial f}{\partial \tilde{x}_{ik}} = -\beta_kg'(\sum_{k=1}^{p}\beta_k\tilde{x}_{ik}), \quad k=1,\dots, p.
\end{equation*}
Since $g$ is strictly increasing, it is obvious that the monotonicity of $f$ with respect to $\tilde{x}_{ik}$ is determined by $\beta_k$.
From Lemma \ref{lem22}, the inverse uncertainty distribution of $\tilde{y}_i-g(\sum_{k=1}^{p}\beta_k\tilde{x}_{ik})$ is derived as 
$
\Psi_{i}^{-1}(\alpha) - g(\sum_{k=1}^{p}\beta_k\varphi_{ik}^{-1}(\alpha)).
$
From Lemma \ref{lem23}, we have
\begin{equation*}
E[\tilde{y}_i - g(\sum_{k=1}^{p}\beta_k\tilde{x}_{ik})]^2 =\int_0^1 ( \Psi_{i}^{-1}(\alpha) - g(\sum_{k=1}^{p}\beta_k\varphi_{ik}^{-1}(\alpha)))^2 d\alpha
\end{equation*}
\\
For $l \geq 2$, substituting $g(\sum_{k=1}^{p}\beta_k\tilde{x}_{ik})$ with $\bm{\delta}^\mathrm{T}(\sum_{k=1}^{p}\beta_k\tilde{x}_{ik})\bm{\hat{b}}$, the proof   is thus complete.
\end{proof}

\begin{remark}
{\rm
For the formula (\ref{eq:17}), if the B-spline basis functions are of degree 1, the resulting fitted curve will be piecewise linear and thus non-differentiable at the knots. In this case, it suffices to ensure that the derivative is positive over each interval between adjacent knots to guarantee global monotonicity. In contrast, for B-splines of degree 2 or higher, the fitted curves are smooth over the entire domain. Therefore, monotonicity can be ensured by requiring that the derivative of the approximating function remains positive everywhere.
}
\end{remark}
When solving the problem $(\ref{eq:17})$, we also adopt the two-step iterative method similar to that used in  Section \ref{Sect:3.1}, obtaining the optimal estimate of $\bm{\beta}$ and $\bm{b}$.
Meanwhile, the selection of the number of basis functions is also crucial. Selecting an appropriate number of basis functions, $m$, is conducive to achieving a good fit for the unknown function while improving computational efficiency. The V-fold cross-validation method can also be used for this purpose, as discussed by Ding and Zhang (\citeyear{DingZhang2021}) and practiced by other scholars.
For each $m$, the $CV(m)$ is mathematically expressed as:
\begin{equation}
{CV}(m) = \frac{1}{V} \sum_{v=1}^{V} \sum_{i \in I_{-v}} \int_0^1( \Psi_{i}^{-1}(\alpha)-	\bm{\delta}^T(\sum_{k=1}^{p}\beta_k\varphi_{ik}^{-1}(\alpha))\bm{\hat{b}})^2 d\alpha,
\label{eq:19}
\end{equation}
where $\bm{\hat{b}}$ is obtained from the $v$th training subset, and $I_{-v}$ stands for the corresponding validation subset. The optimal parameter $m$ is derived by minimizing $CV(m)$, i.e., $\hat{m} = \arg\min_{m} {CV}(m)$.


Given the fitted USIU model derived by Eq.(\ref{eq:15}), the $i$th residual is expressed by
\begin{equation}
\hat{\epsilon}_i=\tilde{y}_i-\bm{\delta}^T(\sum_{k=1}^{p}\hat{\beta}_k\tilde{x}_{ik})\hat{\bm{b}},\quad i = 1, \dots, n.
\label{eq:20}
\end{equation}
Assume that the disturbance terms $\epsilon_1,\ldots,\epsilon_n$ are independent and identically distributed, and that their expectation and variance exist, the expected and variance values of $\hat{\epsilon}$ can be estimated as
\begin{equation*}
\hat{e} = \frac{1}{n} \sum_{i=1}^{n} E\left[\hat{\epsilon}_i\right],
 \quad
\hat{\sigma}^2 = \frac{1}{n} \sum_{i=1}^{n} E\left[(\hat{\epsilon}_i - \hat{e})^2\right].
\end{equation*}
\begin{thm}
\label{thm:3}
{\rm
Under the assumption in Theorem $\ref{thm:2}$, the expected and variance value of the disturbance term $\epsilon$ in USIU model (\ref{eq:17}) is estimated by
\begin{align*}
\hat{e} &= \frac{1}{n} \sum_{i=1}^{n}\int_0^1( \Psi_{i}^{-1}(\alpha)-	\bm{\delta}^T(\sum_{k=1}^{p}\hat{\beta}_k\varphi_{ik}^{-1}(\alpha))\hat{\bm{b}}) d\alpha,
\\
\hat{\sigma}^2 &= \frac{1}{n} \sum_{i=1}^{n}\int_0^1( \Psi_{i}^{-1}(\alpha)-	\bm{\delta}^T(\sum_{k=1}^{p}\hat{\beta}_k\varphi_{ik}^{-1}(\alpha))\hat{\bm{b}}-\hat{e})^2 d\alpha,
\end{align*} 
where $\varphi_{ik}(\alpha)$ is defined in (\ref{varphi}).
}
\end{thm}
\begin{proof} It can be directly from Lemma {\ref{lem23}}.
\end{proof}
Based on the fitted USIU model and following Liu (\citeyear{Liu2026}), the forecast uncertain variable of $\tilde{y}$ respect to new $\tilde{\bm{x}}$ can be written as $\hat{y}=\bm{\delta}^T(\sum_{k=1}^{p}\hat{\beta}_k\tilde{x}_k)\hat{\bm{b}}+\hat{\epsilon}$.
Then, the forecast value $\hat{\mu}$ obtained from the following formula:
\begin{equation}
\hat{\mu}=E[\hat{y}]=\int_0^1\bm{\delta}^T(\sum_{k=1}^{p}\hat{\beta}_k\varphi_{k}^{-1}(\alpha))\hat{\bm{b}}d\alpha+\hat{e}.
\label{eq:21}
\end{equation}
It is further assumed that $\hat{y}$ and $\epsilon$ follow the uncertainty distribution $\hat{\Psi}$ and $\Omega$, and the new explanatory variables $\tilde{x}_1,\ldots,\tilde{x}_p$ follow the regular uncertainty distributions $\Phi_1,\ldots,\Phi_p$. The inverse uncertainty distribution of the predicted uncertain variable $\hat{y}$ is given by
$\hat{\Psi}^{-1}(\alpha) = \bm{\delta}^T(\sum_{k=1}^{p}\hat{\beta}_k\varphi_{k}^{-1}(\alpha))\hat{\bm{b}} + \Omega^{-1}(\alpha)$.
Given the confidence level $\alpha$, let $a^{*}$ be the minimum value such that
\begin{equation}
\hat{\Psi}(\hat{\mu} + a) - \hat{\Psi}(\hat{\mu} - a) \geq \alpha.
\label{eq:22}
\end{equation}
We thus obtain that the confidence interval for the response variable $\hat{y}$ is $[\hat{\mu}-a^{*}, \hat{\mu}+a^{*}]$.

\section{Numerical simulations}\label{sec5}

In this section, two numerical examples illustrate the performance of the proposed USIC and USIU models.
\begin{exam}{\rm
Consider that the explanatory variables $x_1,\ldots,x_p$ are precise, and the response variable $y$ is uncertain. The dataset is generated from the USIC model:
\begin{equation}
\tilde{y}=g(\sum_{k=1}^{3}\beta_kx_{k})+\epsilon,
\label{eq:23}
\end{equation}
where $\epsilon$ is an uncertain disturbance term and follows a linear uncertain distribution. Take $\bm{\beta}=(0.2,-0.4,0.9)^{T}$, and the link function 
$g(z)=(1-z)(z)^3,$
where $z=\sum_{k=1}^{3}\beta_kx_{k}$.
500 samples are generated from the model above.  Only a subset of this dataset is presented in Table \ref{tab:data_subset_1},  and the full dataset is available in the supplementary files. 
In parameter estimation for model (\ref{eq:23}), based on Eq.(\ref{eq:4}), initial values for the coefficient vector $\bm{\beta}$ are computed from the dataset $(\tilde{y}_i, x_{i1}, x_{i2}, x_{i3})$ $(i=1,2,\dots,500)$, with the unknown function $g$ initially set as the identity function, i.e. $g(z)=z$. The function $g$ is then estimated according to formula (\ref{eq:7}) along the fixed direction specified by the current $\bm{\beta}$, and a bound-constrained quasi-Newton method is utilized to execute the iterative procedure introduced in Section \ref{Sect:3.1}.
In this study, the derivative of the standard normal uncertainty distribution, given by formula (\ref{eq:9}), is employed as the kernel function. 

\begin{table}[]
	\captionsetup{
		justification=raggedright, 
		singlelinecheck=off        
	}
	\caption{A subset of data generated from model (\ref{eq:23}).}
	\label{tab:data_subset_1}
	\begin{tabular*}{\textwidth}{@{\extracolsep{\fill}} c r r r l l @{}}
		\toprule
		$i$ & $x_{1i}$ & $x_{2i}$ & $x_{3i}$ & $\tilde{y}_i$ & $\epsilon_i$ \\
		\midrule 		
		75  & -0.0987 & 1.7983 & 0.5606 &
		$\mathcal{L}(-2.0341, 2.0012)$ & $\mathcal{L}(-2.0175, 2.0178)$   \\	
		97  & 1.1292 & 0.0335 & 1.9698 &
		$\mathcal{L}(-7.7007, -7.2921)$ & $\mathcal{L}(-0.2038, 0.2047)$   \\
		162 &-0.5844 & 1.2257 & -1.1705 &
		$\mathcal{L}(-13.8058, -10.2183)$ & $\mathcal{L}(-1.7941, 1.7934)$   \\
		186 & -0.0748 & 0.2025& -1.4600 &
		$\mathcal{L}(-8.4392, -4.7934)$ &  $\mathcal{L}(-1.8234, 1.8224)$  \\
		209 & -1.3174 & -1.2565 & 1.1414 &
		$\mathcal{L}(-3.0453, 2.0054)$ & $\mathcal{L}(-2.5253	, 2.5254)$   \\
		263 & 0.6967 & 0.3017 & -1.8730	 &
		$\mathcal{L}(-12.9676, -11.1840)$ & $\mathcal{L}(-0.8919	, 0.8918)$   \\
		303 & 1.1163 & -1.5149 & 0.1707	 &
		$\mathcal{L}(-0.6989, 0.7325)$ & $\mathcal{L}(-0.7160	, 0.7154)$   \\
		368 & -1.9289 & -1.5430 & 0.3423 &
		$\mathcal{L}(-1.0604, 1.2043)$ & $\mathcal{L}(-1.1322, 1.1325)$   \\
		395 & -1.5783 & 1.2459 & -0.1553 &
		$\mathcal{L}(-3.7716, 0.3933)$ & $\mathcal{L}(-2.0820, 2.0829)$   \\
		486 & -1.5369 & 1.0293 & -1.8067 &
		$\mathcal{L}(-43.4712, -41.4521)$ & $\mathcal{L}(-1.0099, 1.0092)$   \\
		\bottomrule
	\end{tabular*}
\end{table}

For bandwidth selection, we adopt cross-validation in formula (\ref{eq:10}) to evaluate the fitting performance of the uncertain N-W estimator under different bandwidths, and we utilize the Fibonacci search algorithm described in Algorithm \ref{alg:fibonacci_bandwidth} to efficiently determine the optimal bandwidth, which yields $h_{opt}=0.0243$. In addition, the comparision of fitted curves at bandwidths \(0.1\), \(0.3\), \(0.5\), \(h_{\mathrm{opt}}\) against the true function curve are shown in Figure \ref{fig:simulationaplot1}. Using the bandwidth $h_0$,  the estimated single-index coefficients are obtained as $\hat{\beta}_1=0.2025$, $\hat{\beta}_2=-0.3961$, and $\hat{\beta}_3=0.8956$, which are very close to the true parameter vector $\bm{\beta}$ set in this simulation. 

\begin{figure}[H] 
	\centering
	\includegraphics[width=0.6\linewidth]{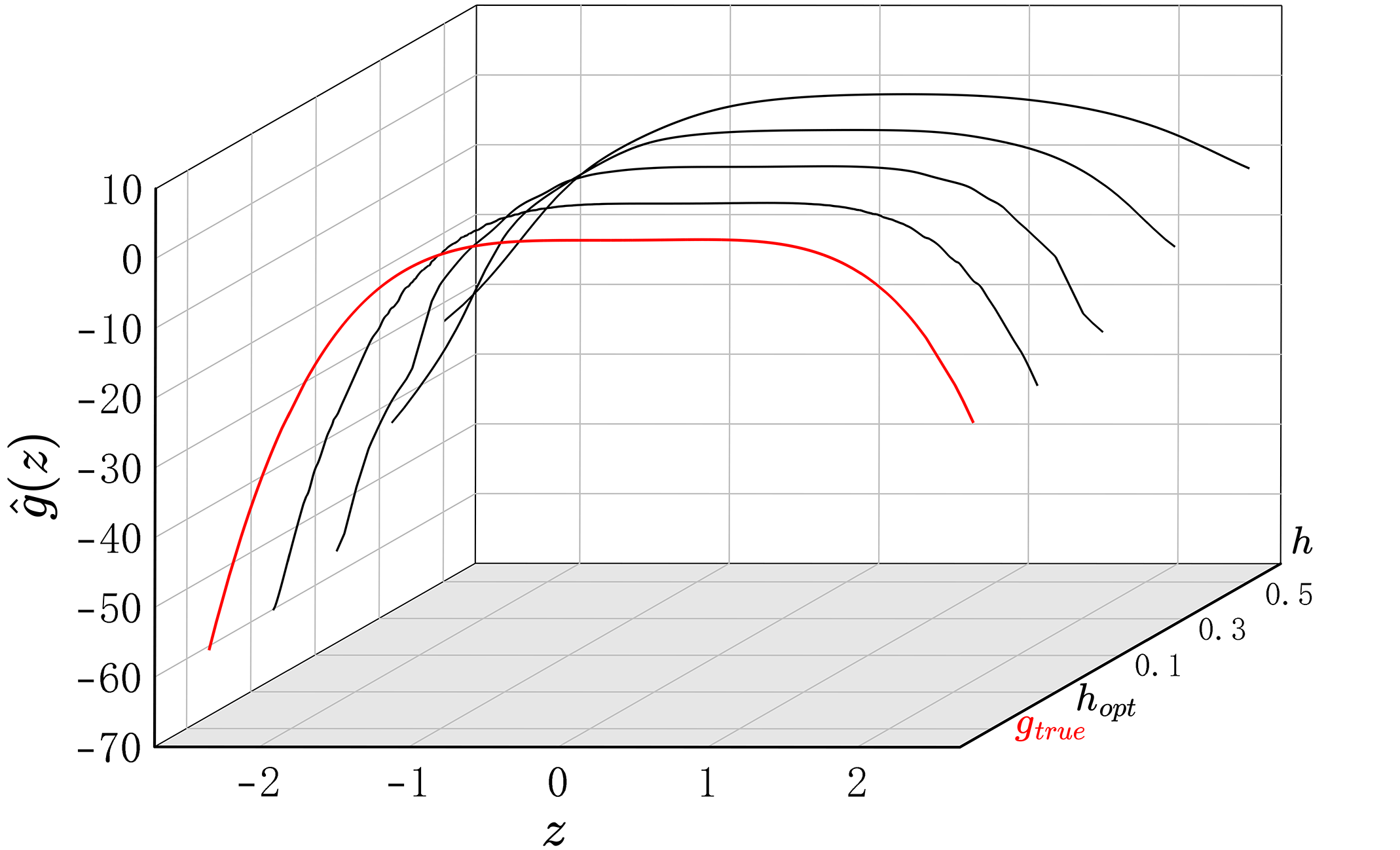}
	\caption{Comparison of fitted curves at bandwidths \(0.1\), \(0.3\), \(0.5\), \(h_{\mathrm{opt}}\) against the true function curve.}
	\label{fig:simulationaplot1}
\end{figure}

Based on the model estimated above, we can calculate the estimates of the expectation and variance of the disturbance as: $\hat{e}=0.0448$, $\hat{\sigma}^2=0.0097$. Further, we assume that the disturbance follows the normal uncertain distribution $\mathcal{N}(0.0448, 0.0985)$. Uncertain hypothesis test (Ye and Liu, \citeyear{YeLiu2022}) can verify the validity of this assumption, and we write it as follows:
\begin{equation*}
H_0: e = 0.0448 \text{ and } \sigma = 0.0985 \quad \text{v.s.} \quad H_1: e \neq 0.0448 \text{ or } \sigma \neq 0.0985.
\end{equation*}
Specify the significance level $\alpha=0.05$, it has $\Theta_{h_0}^{-1}(0.05) = -0.1151$, $\Theta_{h_0}^{-1}(0.95) = 0.2047$, where $\Theta_{h_0}^{-1}$ is the inverse uncertainty distribution of $\mathcal{N}(0.0448, 0.0985)$. Then, 
\begin{equation*}
\begin{split}
W_1 = \bigl\{ \left( \hat{\varepsilon}_1, \hat{\varepsilon}_2, \dots, \hat{\varepsilon}_{500} \right) \mid 
&\text{there are at least } 25 \text{ of indexes } i\text{'s with } 1 \leq i \leq 500 \\
&\hspace{-7em} \text{such that } \hat{\varepsilon}_i \sim \Omega_i \text{ with } \Omega_i^{-1}(0.95) < -0.1151, \text{ or } \Omega_i^{-1}(0.05) > 0.2047 \bigr\},
\end{split}
\end{equation*}
where $\Omega_i^{-1}$ is the inverse distribution of $\varepsilon_i$. Figure \ref{fig:simulationaplot3.3} shows the values of $\Omega_i^{-1}(0.05),\Omega_i^{-1}(0.95)$ for $i = 1, 2, \dots, 500$, with the boundaries of this test. Clearly, only four points fall into the rejection region $W_1$. Therefore, we fail to reject $H_0$. This shows that the model provides a good fit to the sample data.\\
\indent According to formula $(\ref{eq:12})$, given a new set of explanatory variables $(x_1$, $x_2$, $x_3)$= $(2.1$, $-1.5$, $0.4)$, we can obtain the expected value of the new responding variable $\hat{y}$:
\begin{equation*}
	\hat{\mu}=E[\hat{y}]=\hat{g}(\sum_{k=1}^{3}\hat{\beta}_kx_{k})+\hat{e}=-0.9167.
\end{equation*}
At the significance level of $\alpha=0.05$, the minimum value $a$ satisfying the inequality ($\ref{eq:13}$) is $a^{*}=0.0054$. Thus, the confidence interval for response variable $\hat{y}$ is $[\mu-a^{*}, \mu+a^{*}]=[-0.9221, -0.9113]$.


\begin{figure}[H]
	\centering
	\begin{subfigure}{\linewidth}
		\centering
		\setlength{\parskip}{0pt}
		\setlength{\baselineskip}{0pt}
		\caption{}  
		\includegraphics[width=0.8\textwidth,height=6.0cm]{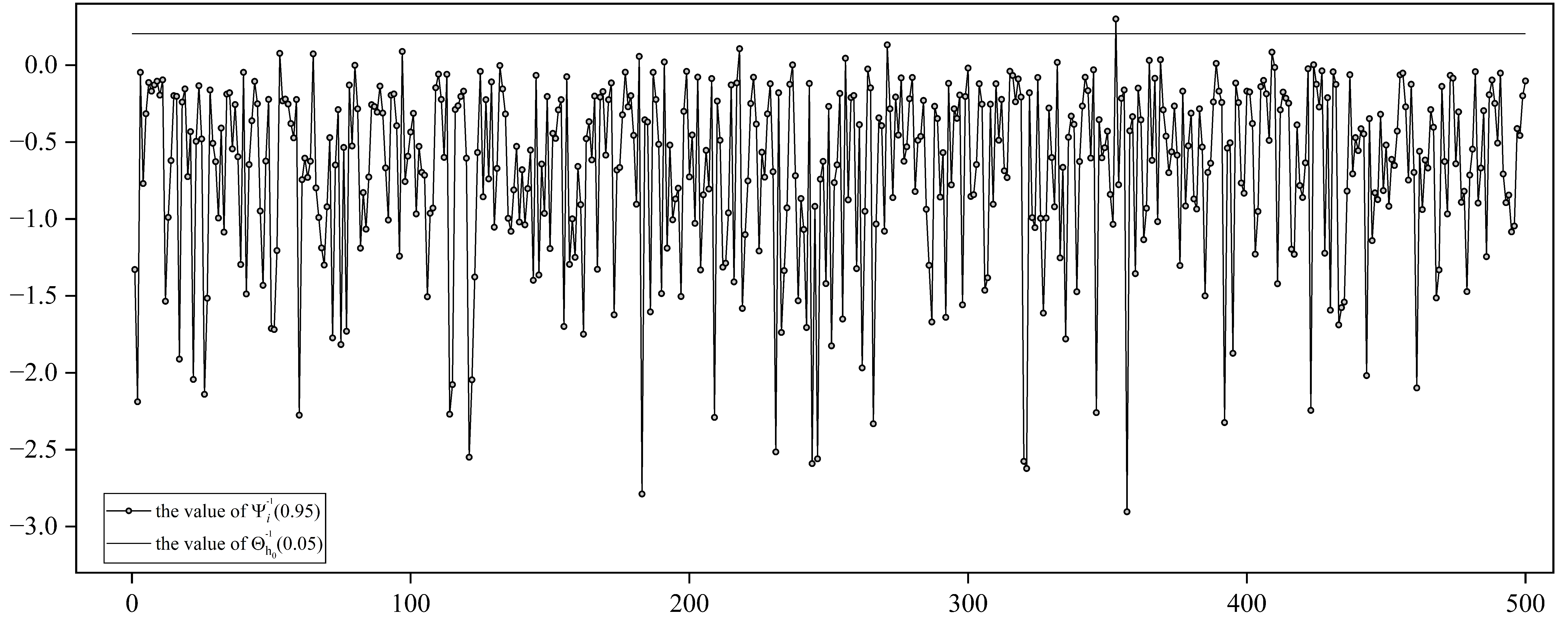} 
		\label{fig:simulationaplot3.1}
	\end{subfigure}
	\begin{subfigure}{\linewidth}
		\centering
		\caption{}  
		\includegraphics[width=0.8\textwidth,height=6.0cm]{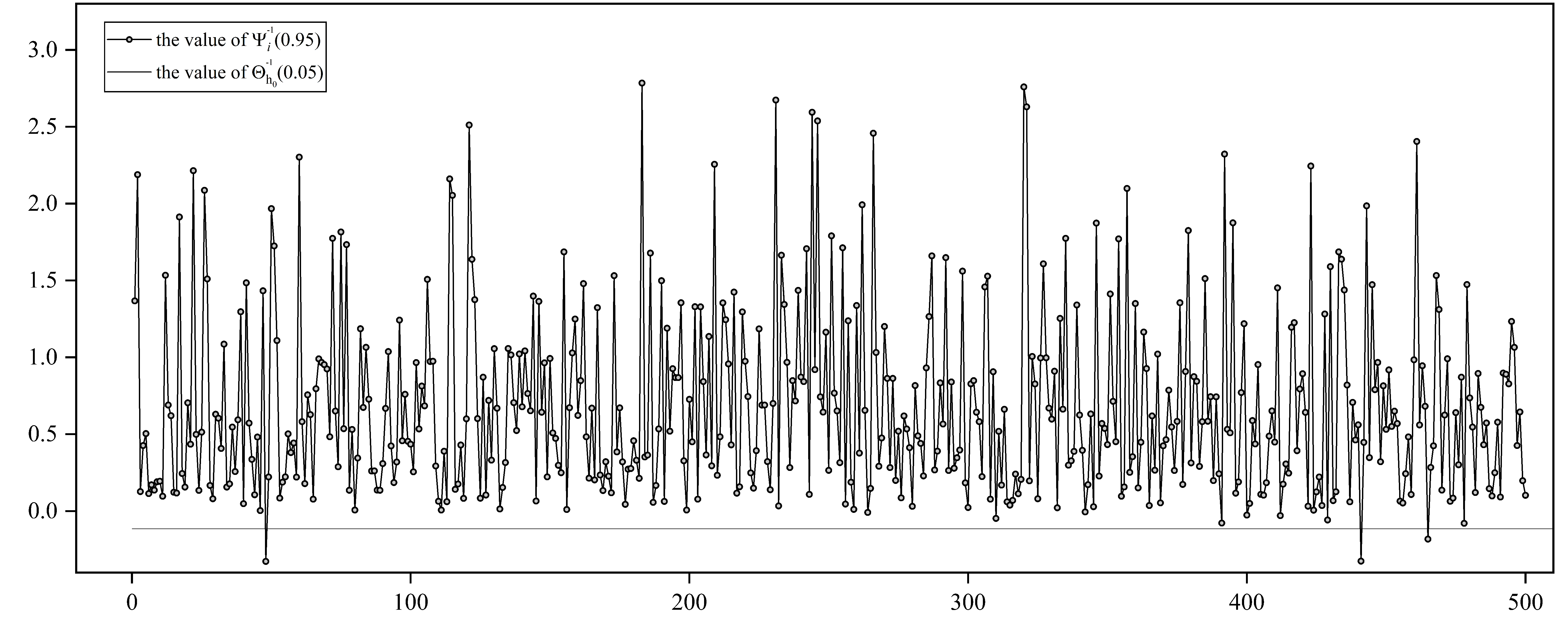} 
		\label{fig:simulationaplot3.2}
	\end{subfigure}
	
	\caption{Scatter plots of $\Omega_i^{-1}(0.05)$ and $\Omega_i^{-1}(0.95)$ for $i = 1, \dots, 500$ with the boundaries of the test $W_1$.}
	\label{fig:simulationaplot3.3}
\end{figure}
}
\end{exam}

\begin{exam}{\rm
Consider both the explanatory variables $x_1,\ldots,x_p$ and the response variable $y$ are imprecise.  In the Table \ref{tab:data_subset_2}, the data are generated from the USIU model as follows:
\begin{equation}
	\tilde{y}=g(\sum_{k=1}^{3}\beta_k\tilde{x}_{k})+\epsilon,
	\label{eq:24}
\end{equation}
where $\epsilon$ is an uncertain disturbance term following a linear uncertain distribution, $\tilde{x}_1$, $\tilde{x}_2$, $\tilde{x}_3$ are mutually independent, and we specify $\beta_1$=0.3, $\beta_2$=0.7, $\beta_3$=0.6481. The data in Table 2 comprises 50 observations.
As in the previous example, we assume the link function is the identity function, thereby obtaining the initial values of the single-index coefficients. 

Based on the fixed $\beta_k(k=1,2,3)$, the B-spline method in Section \ref{Sect:3.2} is employed to obtain the estimate of the unknown function $g$. Notably, cubic B-splines($l=3$) are commonly used to balance the smoothness of fitted curves and computational simplicity, and we follow this convention in this simulation. To arrange appropriate knots when the variable in the B-spline approximate function is imprecise, we suggest first calculating $E[\tilde{t}_i](i=1,\ldots,n)$. Then, the knots can be placed uniformly between the minimum and maximum of these expected values. This interval is updated as $\bm{\beta}$ is renewed. For the cubic B-splines adopted in this paper, we replicate each boundary knot four times to control the B-spline curve so that it starts and ends exactly at the boundaries. Let $E_{min}=\min\{E[\hat{t}_1],\dots,E[\hat{t}_{n}]\},E_{max}=\max\{E[\hat{t}_1],\dots,E[\hat{t}_{n}]\}$, the knots are placed as follows:\\
\[
\left[\bigl\{E_{min},E_{min},E_{min},E_{min}\bigr\},\;
E[\tilde{t}]_{\text{uniformly spaced}},\;
\bigl\{E_{max},E_{max},E_{max},E_{max}\bigr\}\right].
\]
To ensure the approximating function satisfies the monotonicity constraints in Theorem \ref{thm:2}, we impose a monotonicity penalty and verify the monotonicity of the estimated function by plotting its derivative. Then, we use the L-BFGS algorithm in R to implement the two-step iterative procedure described in Section {\ref{Sect:3.2}}, optimizing the model parameters and B-spline coefficients. The optimal number of basis functions is set to $m=8$ via the cross-validation criterion given in Eq.({\ref{eq:19}}). Finally, the expected value of the residual sum of squares is calculated as $\sum_{i=1}^{50} E[(\tilde{y}_i - \bm{\delta}^T(\bm{\hat{\beta}}^T \boldsymbol{\tilde{x}_i})\bm{\hat{b}})^2]=12.411$, and we obtain the optimal single-index coefficients as $\hat{\beta_1}=0.137$, $\hat{\beta_2}=0.709$, $\hat{\beta_3}=0.692$. The estimated B-spline coefficients are given by
\begin{equation*}
\hat{\bm{b}} = \left( 0.046,\ 0.021,\ 0.248,\ 1.296,\ 0.646,\ 0.325,\ 0.124,\ -0.021 \right).
\end{equation*}
The B-spline estimator of the link function is $\bm{\delta}^T(\tilde{t})\hat{\bm{b}}$, and its curve is presented in Figure \ref{fig:simulationbplot1}.\\
\indent Based on the above model, the expectation and variance of the disturbance term are obtained according to Theorem {\ref{thm:3}} as: $\hat{e}=0.009$, $\hat{\sigma}^2=0.038$. We further assume the disturbance follows the normal uncertain distribution $\mathcal{N}(0.009, 0.194)$, an uncertain hypothesis test can also verify the validity of this assumption, and we write it as follows:
\begin{equation*}
	H_0: e = 0.009 \text{ and } \sigma = 0.194 \quad \text{v.s.} \quad H_1: e \neq 0.009 \text{ or } \sigma \neq 0.194.
\end{equation*}
Specify the significance level $\alpha=0.05$, it has $\Theta_{h_0}^{-1}(0.05) = -0.383$, $\Theta_{h_0}^{-1}(0.95) = 0.401$, where $\Theta_{h_0}^{-1}$ is the inverse uncertainty distribution of $\mathcal{N}(-0.383, 0.401)$. Then the test is 
\begin{equation*}
	\begin{split}
		W_2 = \bigl\{ \left( \hat{\epsilon}_1, \hat{\epsilon}_2, \dots, \hat{\epsilon}_{500} \right) \mid 
		&\text{there are at least } 2 \text{ of indexes } i\text{'s with } 1 \leq i \leq 50 \\
		&\hspace{-7em} \text{such that } \hat{\epsilon}_i \sim \Omega_i \text{ with } \Omega_i^{-1}(0.95) < -0.383, \text{ or } \Omega_i^{-1}(0.05) > 0.401 \bigr\},
	\end{split}
\end{equation*}
where $\Omega_i^{-1}$ is the inverse distribution of $\epsilon_i$. Figure \ref{fig:simulationbplot2} shows the values of $\Omega_i^{-1}(0.05),\Omega_i^{-1}(0.95)$ for $i = 1, 2, \dots, 50$, with the boundaries of the test. It is observed that none of the points fall within the rejection region $W_2$. Hence, the null hypothesis cannot be rejected, demonstrating that the above model fits the data well.

\begin{sidewaystable}[htbp]
	\captionsetup{
		justification=raggedright, 
		singlelinecheck=off,        
		skip=0.5pt
	}
	\caption{Data generated from model (24)}
	\label{tab:data_subset_2}
	
	\renewcommand{\arraystretch}{1.7} 
	\scriptsize{\normalsize} 
	
	\begin{tabular*}{\textwidth}{@{\extracolsep{0pt}} 
			l @{\hspace{0.3cm}}                
			*{4}{>{\raggedright\arraybackslash}p{2.3cm}} @{\hspace{0.5cm}}  
			l @{\hspace{0.3cm}}                
			*{4}{>{\raggedright\arraybackslash}p{2.3cm}} @{}  
		}
		\toprule
		$i$ & $\tilde{x}_{1i}$ & $\tilde{x}_{2i}$ & $\tilde{x}_{3i}$ & $\tilde{y}_i$ 
		& $i$ & $\tilde{x}_{1i}$ & $\tilde{x}_{2i}$ & $\tilde{x}_{3i}$ & $\tilde{y}_i$ \\
		\midrule				
		
		1  & $\mathcal{L}(2.053,3.895)$ & $\mathcal{L}(-1.989,4.974)$ & $\mathcal{L}(-0.388,2.707)$ & $\mathcal{L}(1.213,1.216)$ & 
		26 & $\mathcal{L}(-0.661,1.128)$ & $\mathcal{L}(-4.816,-1.055)$ & $\mathcal{L}(-4.701,-4.225)$ & $\mathcal{L}(-1.372,-1.367)$ \\		
		
		2  & $\mathcal{L}(-1.818,4.316)$ & $\mathcal{L}(0.355,4.599)$ & $\mathcal{L}(-0.731,3.699)$ & $\mathcal{L}(1.253,1.259)$ & 
		27 & $\mathcal{L}(-4.064,-3.217)$ & $\mathcal{L}(-1.625,-0.842)$ & $\mathcal{L}(1.246,2.773)$ & $\mathcal{L}(-0.588,-0.569)$  \\
		
		3  & $\mathcal{L}(-2.425,2.863)$ & $\mathcal{L}(2.520,3.111)$ & $\mathcal{L}(-2.120,1.152)$ & $\mathcal{L}(1.040,1.050)$ &
		28 & $\mathcal{L}(-0.217,-2.234)$ & $\mathcal{L}(-3.722,0.533)$ & $\mathcal{L}(-0.289,2.884)$ & $\mathcal{L}(-0.831,-0.232)$ \\
		
		4  & $\mathcal{L}(1.405,2.312)$ & $\mathcal{L}(-1.879,4.977)$ & $\mathcal{L}(0.856,2.744)$ & $\mathcal{L}(1.220,1.237)$ &
		29 & $\mathcal{L}(-3.955,2.882)$ & $\mathcal{L}(-3.916,-1.321)$ & $\mathcal{L}(-3.568,-0.864)$ & $\mathcal{L}(-1.292,-1.281)$ \\
		
		5  & $\mathcal{L}(-0.678,3.865)$ & $\mathcal{L}(1.787,4.774)$ & $\mathcal{L}(0.331,3.069)$ & $\mathcal{L}(1.317,1.320)$ &
		30 & $\mathcal{L}(1.531,3.869)$ & $\mathcal{L}-4.392,-1.988)$ & $\mathcal{L}(-4.210,2.967)$ & $\mathcal{L}(-1.074,-1.065)$ \\
		
		6  & $\mathcal{L}(-1.689,-0.863)$ & $\mathcal{L}(-0.647,3.321)$ & $\mathcal{L}(-0.431,2.298)$ & $\mathcal{L}(0.856,0.861)$ &
		31 & $\mathcal{L}(-3.581,1.091)$ & $\mathcal{L}(-1.376,-0.645)$ & $\mathcal{L}(-3.822,-1.839)$ & $\mathcal{L}(-1.244,-1.237)$ \\
		
		7  & $\mathcal{L}(2.206,4.847)$ & $\mathcal{L}(0.218,0.839)$ & $\mathcal{L}(-2.009,0.427)$ & $\mathcal{L}(0.734,0.748)$ &
		32 & $\mathcal{L}(3.508,4.842)$ & $\mathcal{L}(-4.754,-1.356)$ & $\mathcal{L}(-4.420,0.099)$ & $\mathcal{L}(-1.166,-1.150)$ \\
		
		8 & $\mathcal{L}(3.206,4.545)$ & $\mathcal{L}(-1.815,4.638)$ & $\mathcal{L}(2.442,3.860)$ & $\mathcal{L}(1.332,1.341)$ &
		33  & $\mathcal{L}(-4.115,-0.578)$ & $\mathcal{L}(-4.034,-2.951)$ & $\mathcal{L}(-3.446,3.267)$ & $\mathcal{L}(-1.275,-1.262)$ \\
		
		9 & $\mathcal{L}(-0.025,2.222)$ & $\mathcal{L}(4.171,4.652)$ & $\mathcal{L}(0.991,3.806)$ & $\mathcal{L}(1.366,1.378)$ &
		34 & $\mathcal{L}(-2.688,4.399)$ & $\mathcal{L}(-4.007,1.257)$ & $\mathcal{L}(-0.944,2.997)$ & $\mathcal{L}(-0.048,-0.032)$ \\
		
		10  & $\mathcal{L}(-1.100,0.215)$ & $\mathcal{L}(-1.978,3.575)$ & $\mathcal{L}(-1.280,1.073)$ & $\mathcal{L}(0.337,0.352)$ &
		35  & $\mathcal{L}(-4.277,-2.838)$ & $\mathcal{L}(-3.378,1.596)$ & $\mathcal{L}(-2.443,-0.014)$ & $\mathcal{L}(-1.198,-1.178)$  \\
		
		11  & $\mathcal{L}(0.646,0.968)$ & $\mathcal{L}(-3.293,2.464)$ & $\mathcal{L}(2.885,4.403)$ & $\mathcal{L}(1.159,1.168)$ &
		36  &$\mathcal{L}(0.603,1.556)$ & $\mathcal{L}(-1.713,1.128)$ & $\mathcal{L}(-3.328,-2.739)$ & $\mathcal{L}(-1.078,-1.072)$ \\
		
		12  & $\mathcal{L}(2.162,3.803)$ & $\mathcal{L}(0.987,1.318)$ & $\mathcal{L}(4.191,4.887)$ & $\mathcal{L}(1.357,1.361)$ &
		37  & $\mathcal{L}(1.699,2.294)$ & $\mathcal{L}(-4.530,1.211)$ & $\mathcal{L}(-4.264,3.367)$ & $\mathcal{L}(-0.708,-0.706)$ \\
		
		13 & $\mathcal{L}(-2.762,3.398)$ & $\mathcal{L}(2.338,4.616)$ & $\mathcal{L}(-4.975,-0.865)$ & $\mathcal{L}(0.557,0.577)$ &
		38 & $\mathcal{L}(-1.481,0.583)$ & $\mathcal{L}(-4.932,-4.066)$ & $\mathcal{L}(-2.789,-1.178)$ & $\mathcal{L}(-1.364,-1.346)$ \\
		
		14 & $\mathcal{L}(-3.477,-0.118)$ & $\mathcal{L}(-0.967,-0.714)$ & $\mathcal{L}(1.195,4.686)$ & $\mathcal{L}(0.655,0.666)$  &
		39 & $\mathcal{L}(1.623,4.689)$ & $\mathcal{L}(-3.517,1.541)$ & $\mathcal{L}(-3.172,-1.783)$ & $\mathcal{L}(-0.935,-0.933)$ \\
		
		15  & $\mathcal{L}(-2.428,4.745)$ & $\mathcal{L}(-2.222,4.993)$ & $\mathcal{L}(2.353,3.267)$ & $\mathcal{L}(1.262,1.263)$ &
		40 & $\mathcal{L}(0.765,4.921)$ & $\mathcal{L}(-4.801,-3.559)$ & $\mathcal{L}(-4.344,-0.873)$ & $\mathcal{L}(-1.320,-1.301)$ \\
		
		16  & $\mathcal{L}(-3.629,-2.279)$ & $\mathcal{L}(3.195,3.966)$ & $\mathcal{L}(2.323,3.528)$ & $\mathcal{L}(1.287,1.300)$ &
		41 & $\mathcal{L}(-2.741,-1.317)$ & $\mathcal{L}(-4.812,-3.895)$ & $\mathcal{L}(-4.907,0.493)$ & $\mathcal{L}(-1.385,-1.369)$ \\
		
		17 & $\mathcal{L}(-4.363,-2.470)$ & $\mathcal{L}(0.208,1.727)$ & $\mathcal{L}(-0.006,4.989)$ & $\mathcal{L}(0.894,0.911)$ &
		42 & $\mathcal{L}(-3.434,4.494)$ & $\mathcal{L}(-0.115,0.687)$ & $\mathcal{L}(-4.266,-1.670)$ & $\mathcal{L}(-1.007,-0.996)$ \\
		
		18 & $\mathcal{L}(3.538,4.950)$ & $\mathcal{L}(1.887,1.951)$ & $\mathcal{L}(0.259,3.061)$ & $\mathcal{L}(1.299,1.314)$ &
		43 	& $\mathcal{L}(-0.341,3.028)$ & $\mathcal{L}(-3.608,3.183)$ & $\mathcal{L}(-3.924,-3.694)$ & $\mathcal{L}(-1.152,-1.142)$ \\
		
		19  & $\mathcal{L}(2.040,3.695)$ & $\mathcal{L}(1.985,4.910)$ & $\mathcal{L}(-1.429,3.802)$ & $\mathcal{L}(1.323,1.334)$ &
		44  & $\mathcal{L}(-3.629,-2.260)$ & $\mathcal{L}(0.426,1.097)$ & $\mathcal{L}(-2.694,0.603)$ & $\mathcal{L}(-0.805,-0.794)$ \\
		
		20 & $\mathcal{L}(1.601,1.974)$ & $\mathcal{L}(-4.013,2.038)$ & $\mathcal{L}(2.151,3.227)$ & $\mathcal{L}(1.005,1.012)$ &
		45 & $\mathcal{L}(-4.556,1.912)$ & $\mathcal{L}(-4.008,-3.832)$ & $\mathcal{L}(-4.511,-3.624)$ & $\mathcal{L}(-1.401,-1.399)$  \\
		
		21 & $\mathcal{L}(-4.729,-1.857)$ & $\mathcal{L}(1.665,2.769)$ & $\mathcal{L}(-0.617,-0.463)$ & $\mathcal{L}(0.204,0.217)$ &
		46 & $\mathcal{L}(-2.556,-0.639)$ & $\mathcal{L}(-0.504,1.062)$ & $\mathcal{L}(-3.110,2.478)$ & $\mathcal{L}(-0.457,-0.452)$ \\
		
		22  & $\mathcal{L}(0.623,3.911)$ & $\mathcal{L}(1.057,1.980)$ & $\mathcal{L}(-0.373,1.916)$ & $\mathcal{L}(1.144,1.159)$ &
		47 & $\mathcal{L}(-4.330,3.665)$ & $\mathcal{L}(-4.848,-4.071)$ & $\mathcal{L}(-3.332,2.395)$ & $\mathcal{L}(-1.297,-1.292)$ \\
		
		23  & $\mathcal{L}(-0.259,3.901)$ & $\mathcal{L}(4.626,4.880)$ & $\mathcal{L}(2.981,3.652)$ & $\mathcal{L}(1.404,1.409)$ &
		48 & $\mathcal{L}(-3.753,-3.448)$ & $\mathcal{L}(-0.741,-0.494)$ & $\mathcal{L}(-4.199,-3.812)$ & $\mathcal{L}(-1.339,-1.326)$ \\
		
		24 & $\mathcal{L}(0.924,2.065)$ & $\mathcal{L}(-2.728,2.628)$ & $\mathcal{L}(0.473,4.090)$ & $\mathcal{L}(1.075,1.094)$ &
		49 & $\mathcal{L}(-3.347,-2.873)$ & $\mathcal{L}(-3.600,-3.008)$ & $\mathcal{L}(-3.101,-0.684)$ & $\mathcal{L}(-1.357,-1.344)$ \\
		
		25 & $\mathcal{L}(-1.406,-1.273)$ & $\mathcal{L}(1.754,4.805)$ & $\mathcal{L}(3.824,4.151)$ & $\mathcal{L}(1.342,1.360)$ &
		50 & $\mathcal{L}(-0.711,1.230)$ & $\mathcal{L}(-3.087,-2.429)$ & $\mathcal{L}(-3.629,0.821)$ & $\mathcal{L}(-1.230,-1.218)$ \\
		\bottomrule
	\end{tabular*}
	\normalsize
\end{sidewaystable}

Based on the Eq.(\ref{eq:21}), given a new set of uncertain variables $\tilde{x}_1\sim\mathcal{L}(-2, -1)$, $\tilde{x}_2\sim\mathcal{L}(0, 2)$, $\tilde{x}_3\sim\mathcal{L}(2, 3)$, the expected value of new response $\hat{y}$ is
\begin{equation*}
\hat{\mu}=E[\hat{y}]=\int_0^1\bm{\delta}^T(2.245\alpha+1.11)\hat{\bm{b}} d\alpha+\hat{e}=0.902.
\end{equation*}
At the significance level of $\alpha=0.05$, the minimum value $a$ satisfying the inequality ($\ref{eq:22}$) is $a^{*}=0.460$. Thus, the confidence interval for response variable $\hat{y}$ is $[\mu-a^{*}, \mu+a^{*}]=[0.442, 1.362]$.

\begin{figure}[] 
	\centering
	\includegraphics[width=0.5\linewidth]{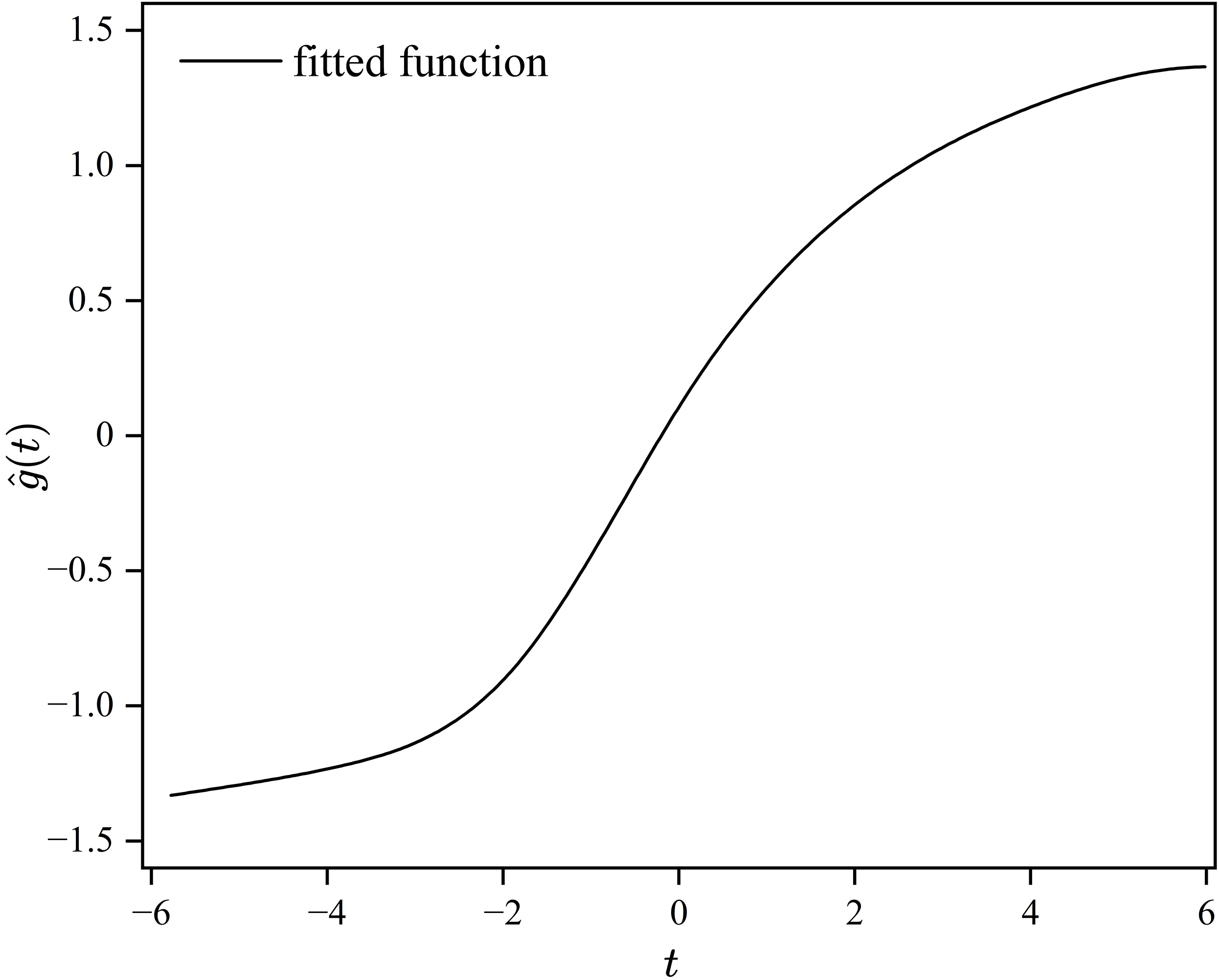}
	\caption{ The curves of estimated link function in the USIU model}
	\label{fig:simulationbplot1}
\end{figure}

\begin{figure}[H]
	\centering
	\begin{subfigure}{0.48\textwidth}
		\centering
		\caption{}
		\includegraphics[width=1\textwidth]{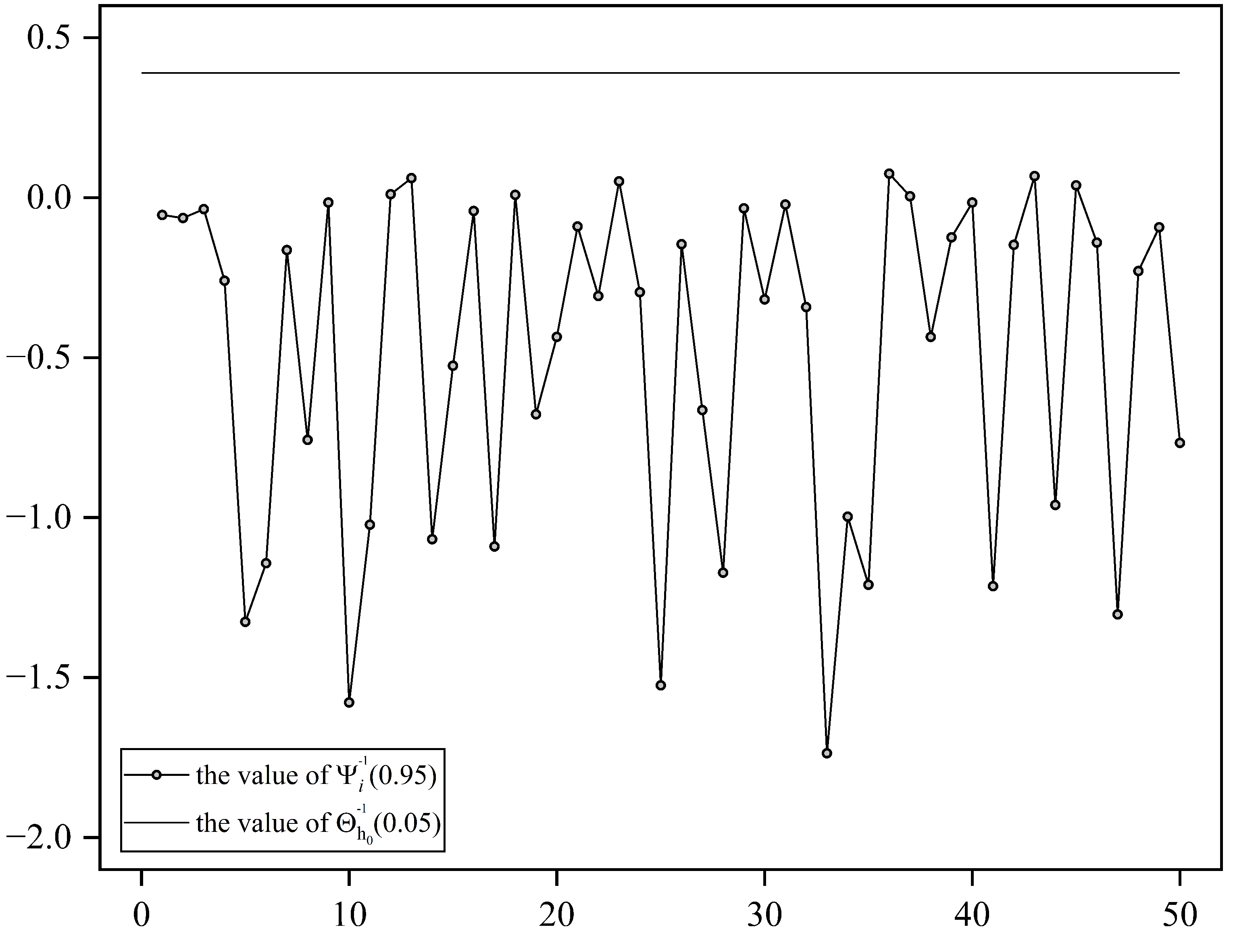} 
		\label{fig:simulationbplot2.1}
	\end{subfigure}
	\hfill  
	\begin{subfigure}{0.48\textwidth}
		\centering
		\caption{}
		\includegraphics[width=1\textwidth]{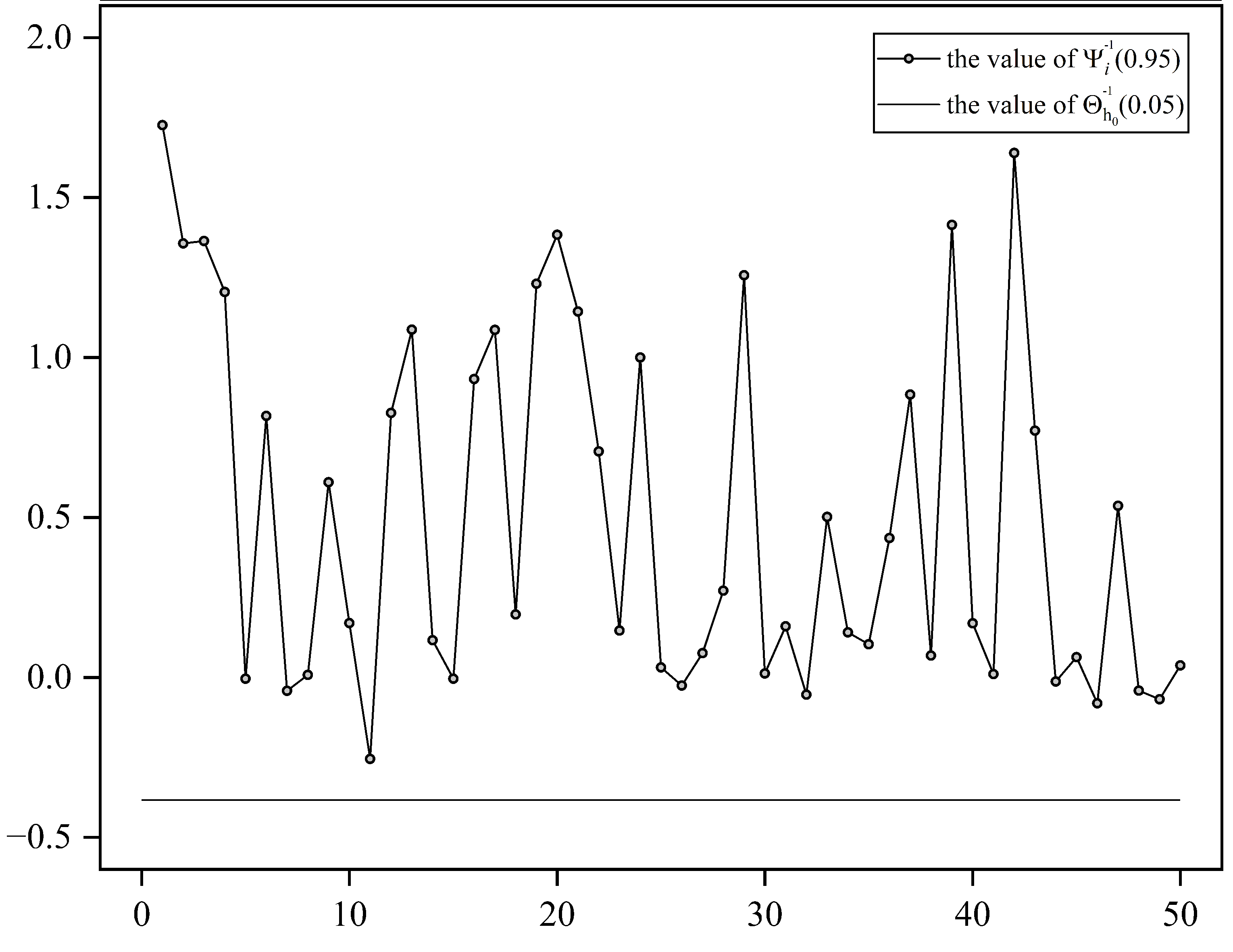} 
		\label{fig:simulationbplot2.2}
	\end{subfigure}
	
	\caption{Scatter plots of $\Omega_i^{-1}(0.05),\Omega_i^{-1}(0.95)$ for $i = 1, \dots, 50$ with the boundaries of the test $W_2$.}
	\label{fig:simulationbplot2}
\end{figure}
}
\end{exam}

\section{An application}\label{sec6}
In this section, a weather dataset of Lagos, Nigeria, is adopted to illustrate the practical applicability of the proposed model. The data are collected from the Kaggle platform (https://www.kaggle.com/). To eliminate interference from long-term climate trends and ensure data uniformity, the study period is set from January 1, 2023, to July 3, 2024, and the dataset contains 535 observations after outlier removal.  
In this case, the USIC model is employed:
\begin{equation}
\tilde{y}=g(\beta_1x_1+\beta_2x_2+\beta_3x_3+\beta_4x_4+\beta_5x_5)+\epsilon.
\label{eq:25}
\end{equation}
In the model (\ref{eq:25}), the response variable $y$ denotes the daily temperature range, and the explanatory variables $x_i (i=1,\ldots,5)$ consist of precipitation, wind speed, humidity, sea-level air pressure, and cloud coverage. The detailed definitions of all variables, together with examples of original and standardized data, are presented in Table \ref{tab:data_subset_3}.
Based on samples $(\tilde{y}_i,x_{1i},x_{2i},x_{3i},x_{4i},x_{5i})$ for $i=1,\ldots,535$, preliminary estimates of the index coefficients are obtained by specifying the initial fitting form of the unknown link function $g$. Taking the estimates as the initial values of $\bm{\beta}$, following the two-step iterative step in Sect.\ref{Sect:3.1}, we derive the final estimators as: $\beta_1=0.1770$, $\beta_2=0.0381 $, $\beta_3=0.8952$, $\beta_4=0.3438$, $\beta_5=0.2183$. In this iteration, the derivative kernel for the uncertain standard normal distribution in Equation (\ref{eq:9}) is adopted. The minimum cross-validation error $CV_{min}=0.5929$ is achieved at bandwidth $h=0.2026$ using Algorithm \ref{alg:fibonacci_bandwidth}, with the corresponding standard error as $0.0398$. To ensure sufficient smoothness and stability of the fitted curve, we adopt a one-standard-error rule (\citeyear{breiman1984cart}, \citeyear{hastie2009elements}) and select the largest bandwidth such that its cross-validation error is within one standard error of the minimum. The resulting modified bandwidth is $h^*=0.5461$, with the cross-validation curve shown in Figure \ref{fig:application1} (a). The final fitting curve of the link function is presented in Figure \ref{fig:application1} (b). 

\begin{table}[H]
	\captionsetup{
		justification=raggedright, 
		singlelinecheck=off        
	}
	\caption{Variable definitions with original and standardized data illustrations.}
	\label{tab:data_subset_3}
	\small  
	\begin{tabular*}{\textwidth}{@{\extracolsep{\fill}} c l l c c @{}}
		\toprule
		Notation & Variable & Definition & Original & Standardized \\
		\midrule 		
		$a_y$  & Tempmin  & Minimum air temperature (℃) & 20.8 &  -  \\	
		$b_y$  & Tempmax  & Maximum air temperature (℃) & 35   & -   \\	
		$x_1$  & Precip   & Total precipitation (mm)     & 13   & 0.62   \\	
		$x_2$  & Humidity & Average relative humidity (\%)& 86.9   & 0.83 \\	
		$x_3$  & Windspeed & Average wind speed (m/s)& 18.9   & -0.45\\	
		$x_4$  & Sealevelpressure & Average sea-level pressure (hPa)& 1014.2   & 1.32 \\	
		$x_5$  & Cloudcover & Average cloud cover ratio (\%)& 59  & 0.45 \\	
		\bottomrule
	\end{tabular*}
\end{table}

 The single-index weights indicate that humidity ($\beta_3•$=0.8952) acts as the dominant factor governing temperature variations in Lagos. In tropical coastal regions, high humidity strongly regulates thermal conditions through dual mechanisms of evaporative cooling and longwave radiation trapping. Sealevel pressure ($\beta_4$=0.3438) and cloud cover ($\beta_5$=0.2831) serve as secondary regulating factors, indirectly affecting temperature by modulating atmospheric circulation and the surface radiation budget, respectively. Precipitation ($\beta_1$=0.1770) and windspeed ($\beta_2$=0.0381) exhibit limited contributions. Precipitation mainly occurs as short-term convective rainfall, while windspeed remains stable under persistent land–sea breezes, resulting in weak independent explanatory power for temperature dynamics.
The estimated link function curve presents a distinct unimodal pattern. As shown in Figure \ref{fig:application_all}, when the single-index value is lower than the approximate threshold $z=-1.3$, it corresponds to the hot and clear conditions of the dry season, where shortwave radiation heating dominates, and temperature rises gradually with increasing $z$. When $z$ exceeds the threshold, the climate shifts to a wet-season convective regime. The joint effects of high humidity, extensive cloud cover, and rainfall induce evaporative cooling, cloud shading, and cold-pool effects, shifting the surface energy balance from radiation to latent heat dominance and leading to a rapid temperature decline.

 In summary, the model captures the nonlinear relationship between air temperature and five other meteorological variables. Moreover, it confirms that temperature regulation over Lagos switches between two distinct dynamic regimes across the threshold. Specifically, the system shifts from a radiation-dominated state to one controlled by latent heat processes, while the key driving factors remain unchanged.


\begin{figure}[]
	\centering
	\begin{subfigure}{0.48\textwidth}
		\centering
		\caption{}
		\includegraphics[width=1\textwidth]{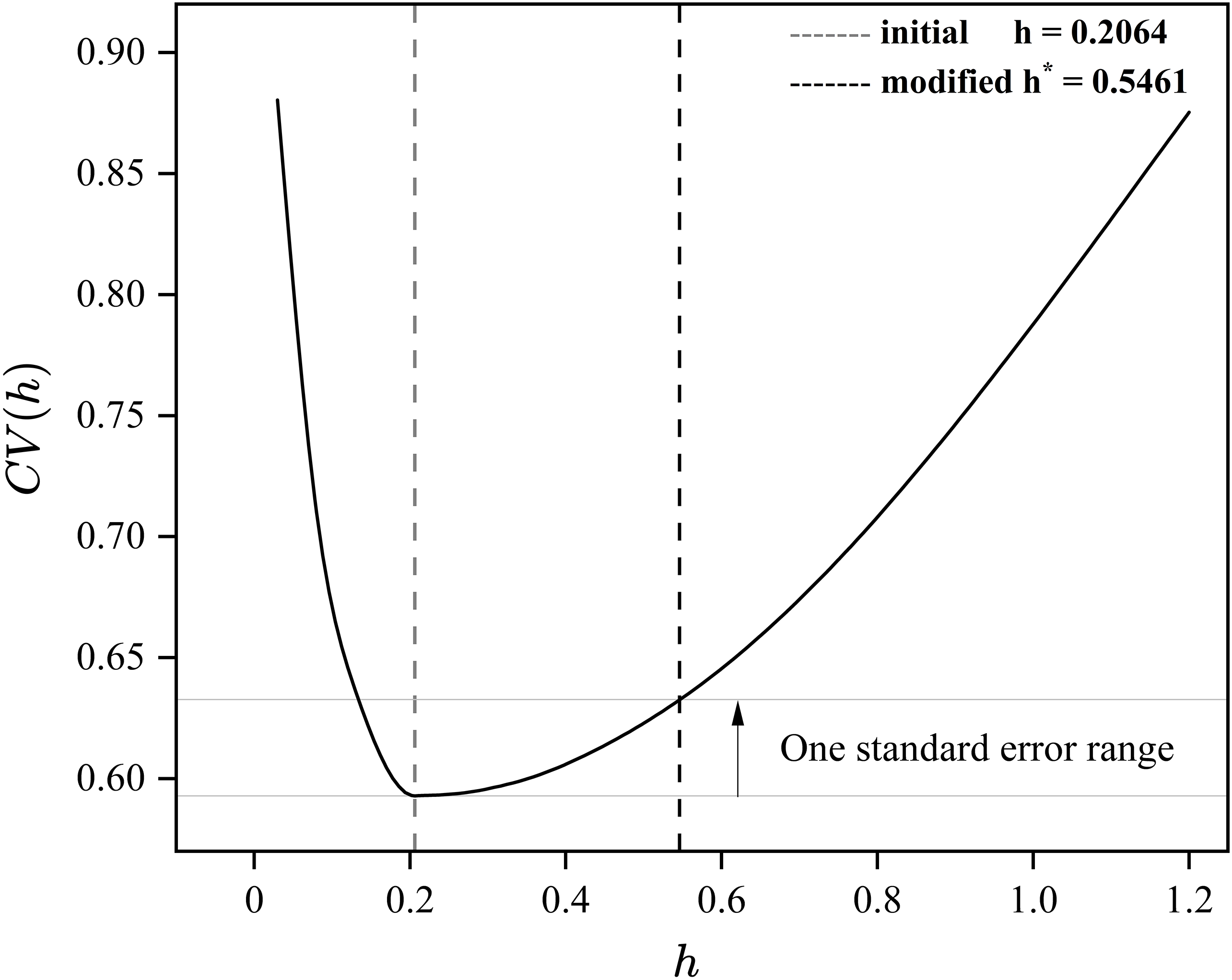} 
		\label{fig:application1.1}
	\end{subfigure}
	\hfill  
	\begin{subfigure}{0.48\textwidth}
		\centering
		\caption{}
		\includegraphics[width=1\textwidth]{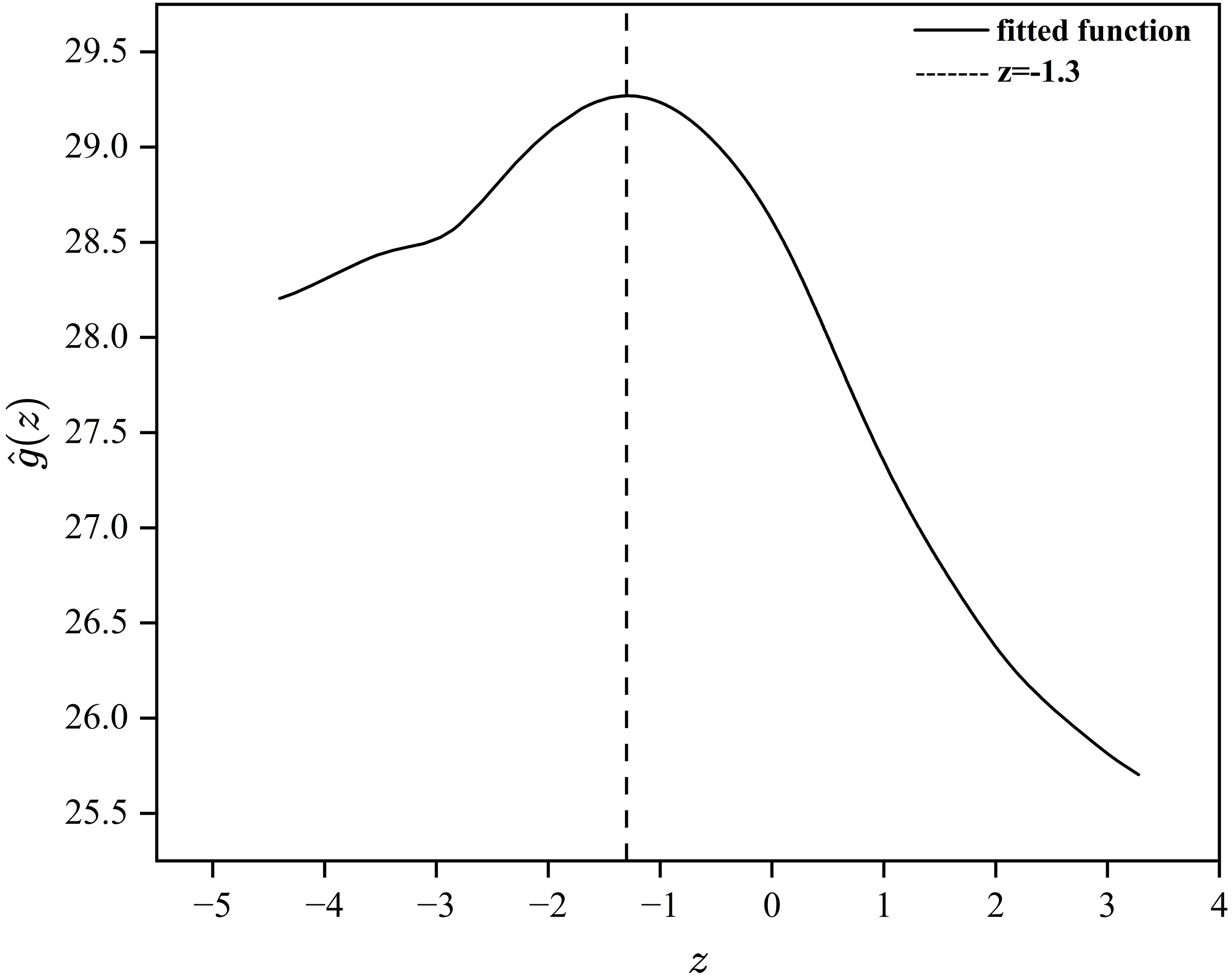} 
		\label{fig:application1.2}
	\end{subfigure}
	
	\caption{(a) Cross-validation curve for bandwidth selection. (b) Plot of the estimated link function.}
	\label{fig:application1}
\end{figure}

\begin{figure}[htbp]
	\centering
	\def\figwidth{0.48\textwidth}
	
	\begin{subfigure}{\figwidth}
    	\centering
    	\includegraphics[width=\linewidth]{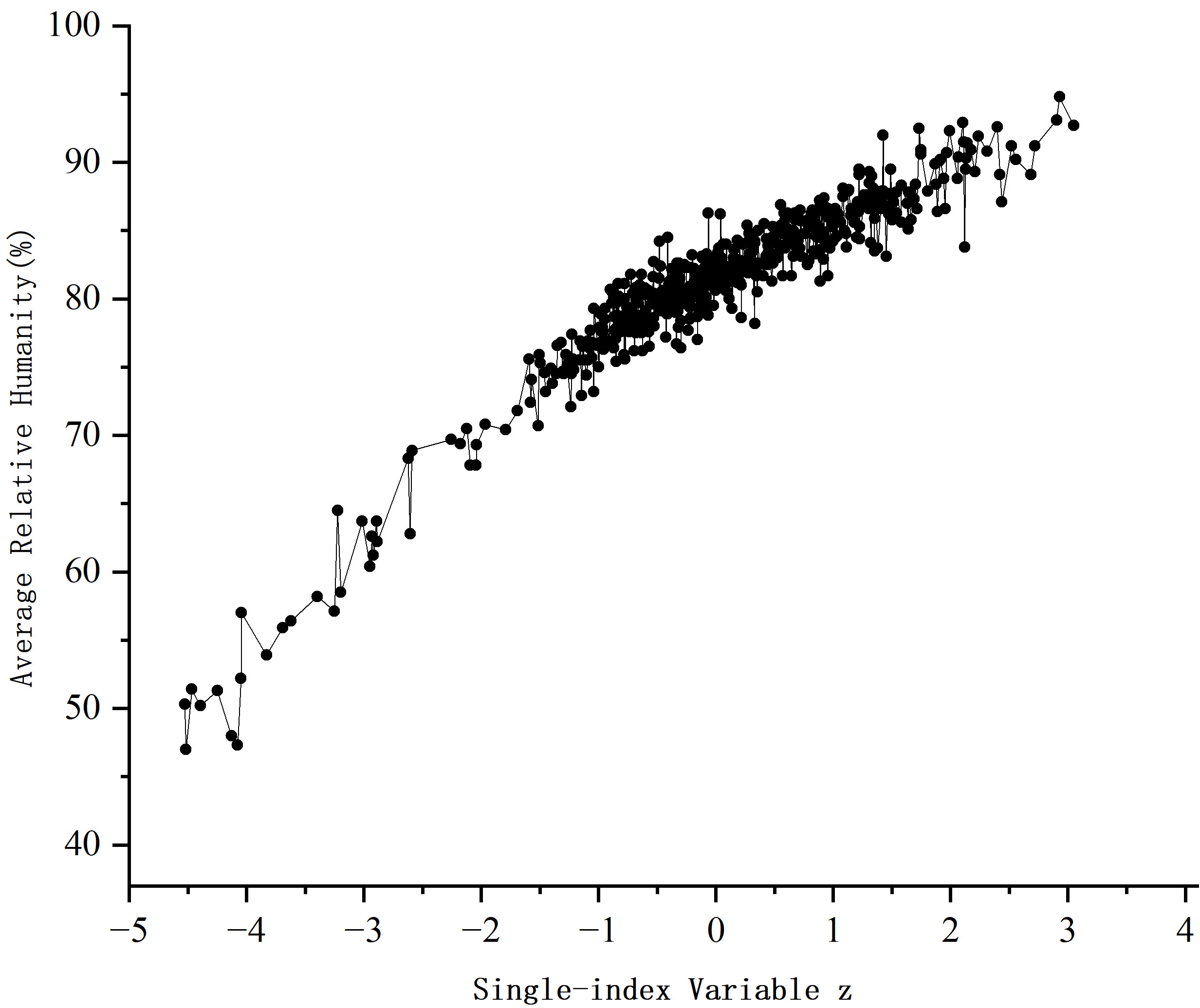} 

    	\caption{}
	\end{subfigure}
	
	\vspace{12pt}
	
	\begin{subfigure}{\figwidth}
		\centering
		\includegraphics[width=\linewidth]{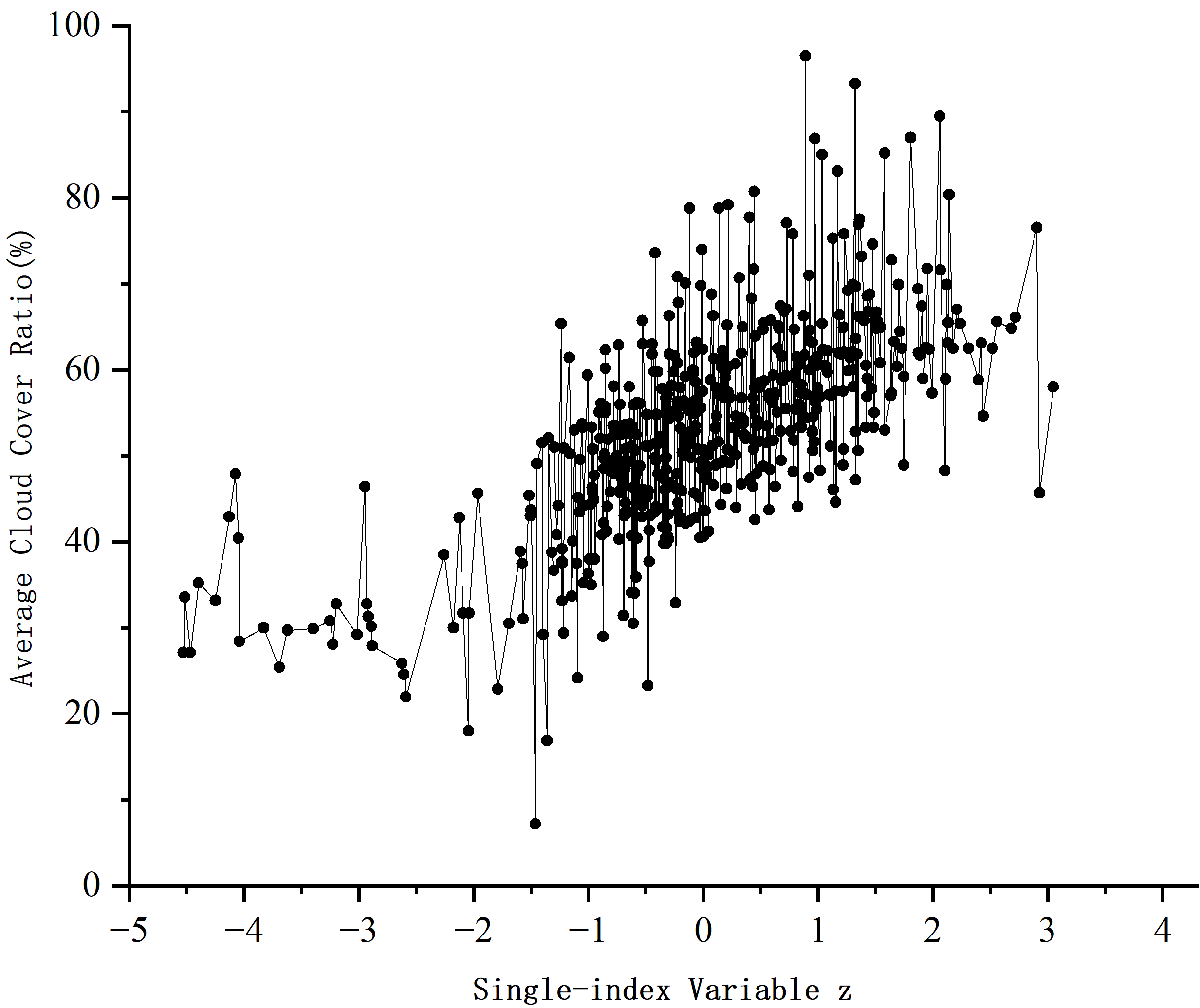}
		\caption{}
	\end{subfigure}
	\hfill
	\begin{subfigure}{\figwidth}
		\centering
		\includegraphics[width=\linewidth]{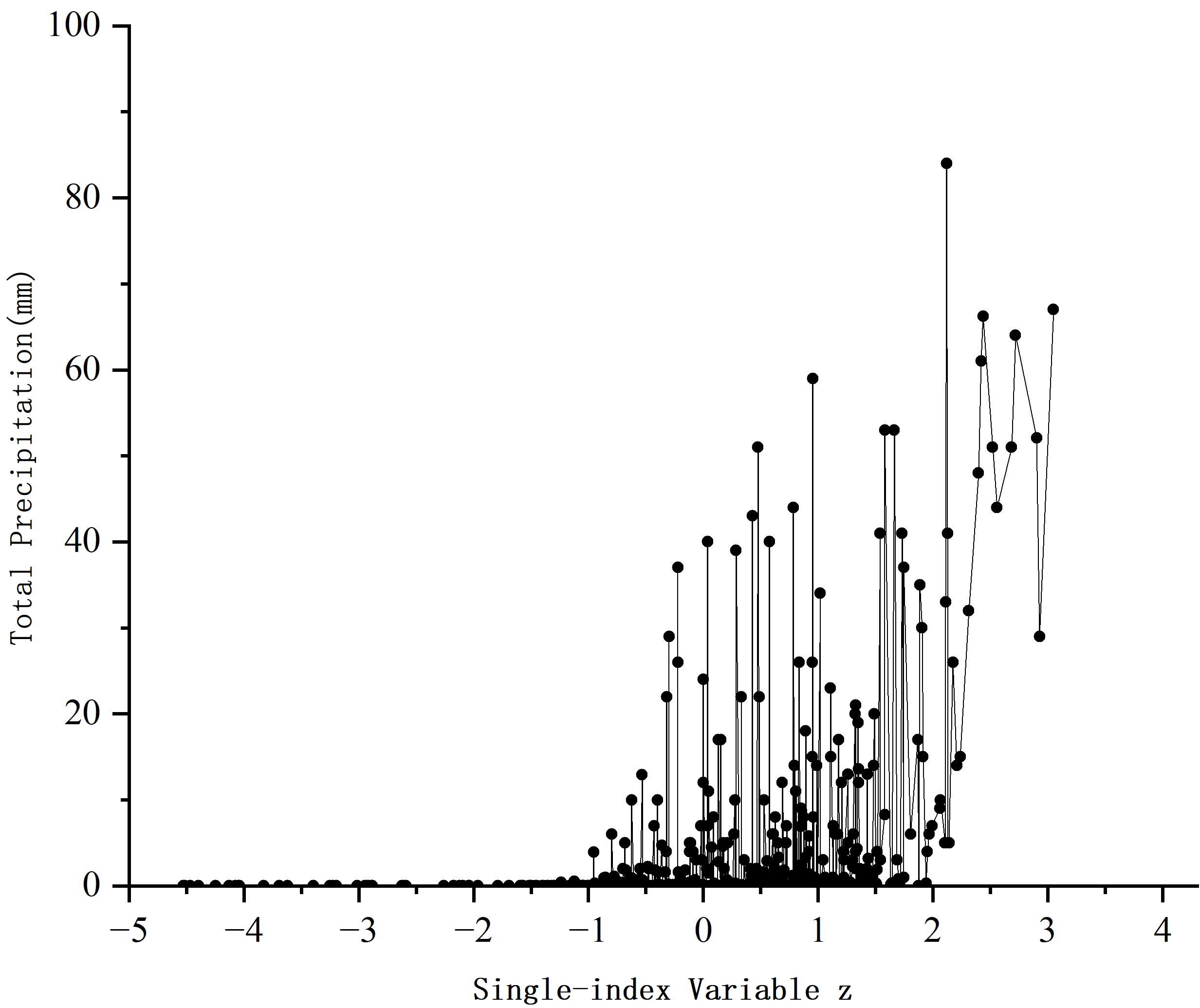}
		\caption{}
	\end{subfigure}
	
	\vspace{12pt}
	
	\begin{subfigure}{\figwidth}
		\centering
		\includegraphics[width=\linewidth]{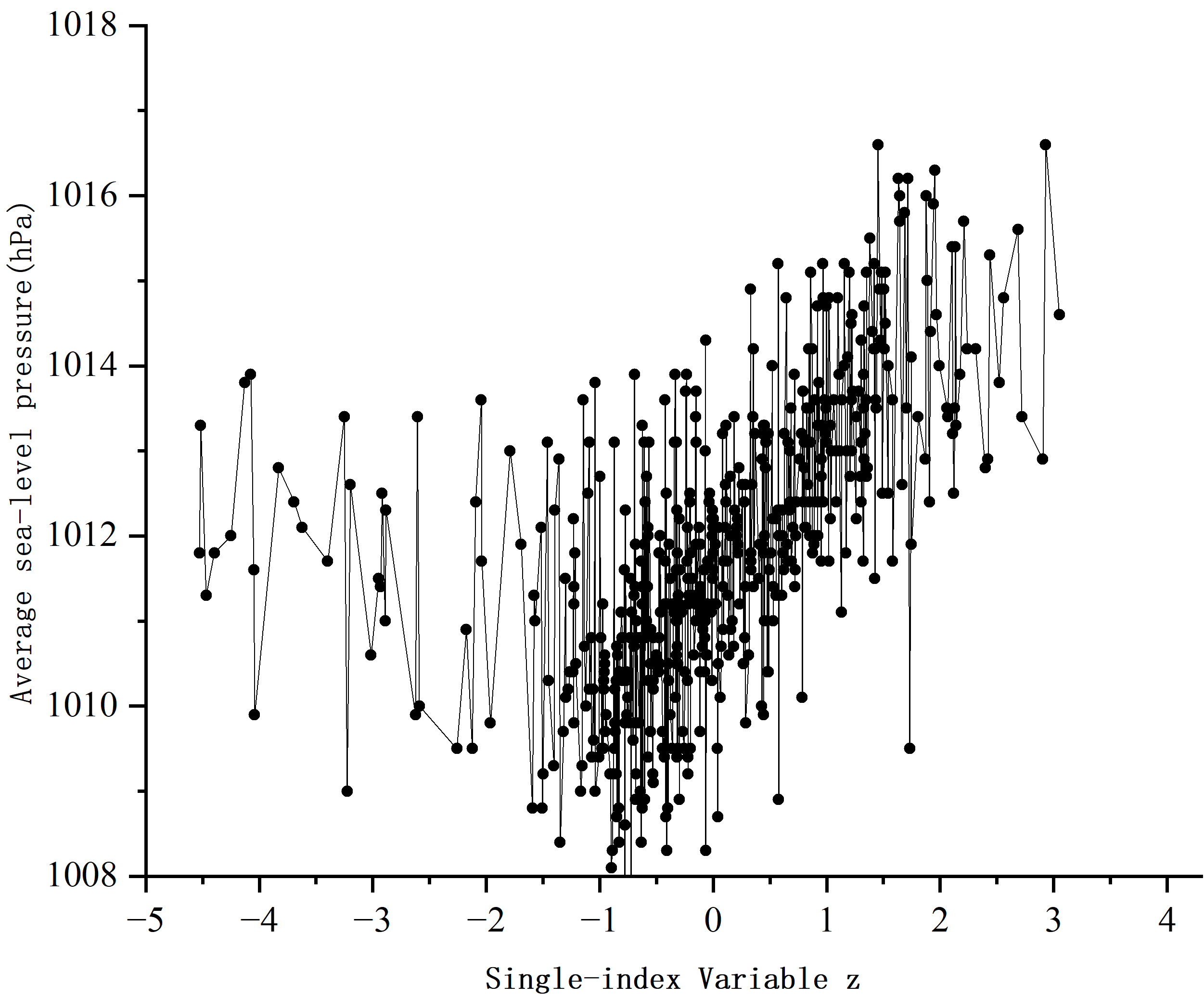}
		\caption{}
	\end{subfigure}
	\hfill
	\begin{subfigure}{\figwidth}
		\centering
		\includegraphics[width=\linewidth]{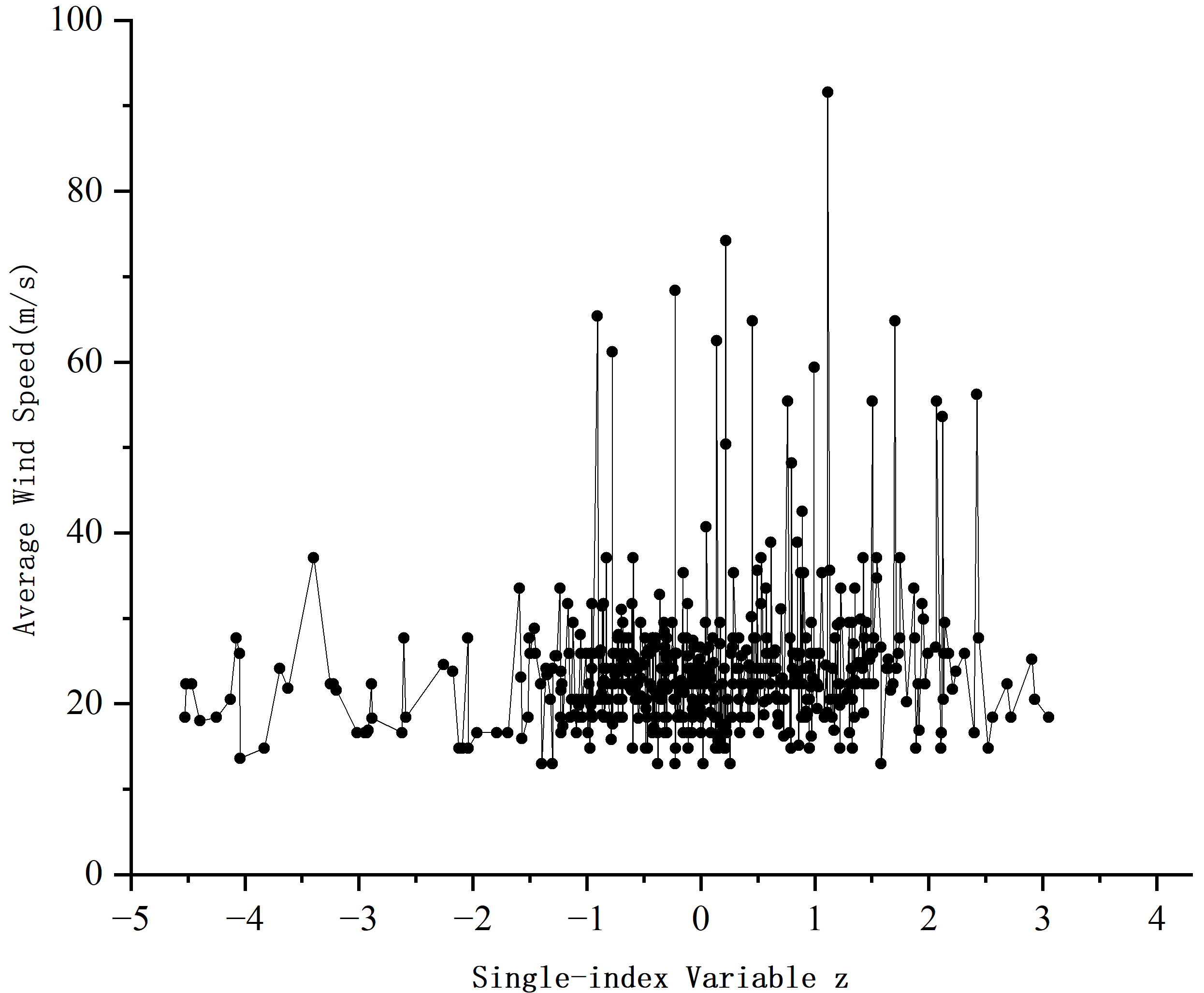}
		\caption{}
	\end{subfigure}
	
	\vspace{5pt}
	\caption{Relationships between key meteorological variables and the single-index value in Lagos.}
	\label{fig:application_all}
\end{figure}

\subsection{Residual analysis}
Next, we conduct residual analysis on the fitted model. By Eq.(\ref{eq:12}), we calculate the expectation and variance of the disturbance term: $\hat{e}=-0.0509$, $\hat{\sigma}^2=0.9765$. Following the procedure of residual analysis, we assume that the disturbance term follows a normal uncertain distribution $\mathcal{N}(-0.0509, 0.9882)$. A hypothesis test is further conducted to verify the rationality of this assumption.The hypotheses are formulated as:
\begin{equation*}
	H_0: e =-0.0509 \text{ and } \sigma = 0.9882 \quad \text{v.s.} \quad H_1: e \neq -0.0509 \text{ or } \sigma \neq 0.9882.
\end{equation*}
Let the significance level $\alpha=0.05$, we have $\Theta_{h_0}^{-1}(0.05) = -1.6551$, $\Theta_{h_0}^{-1}(0.95) = 1.5533$, where $\Theta_{h_0}^{-1}$ is the inverse uncertainty distribution of $\mathcal{N}(-0.0509, 0.9882)$. The rejection region is  
\begin{equation*}
	\begin{split}
		W_3 = \bigl\{ \left( \hat{\epsilon}_1, \hat{\epsilon}_2, \dots, \hat{\epsilon}_{535} \right) \mid 
		&\text{there are at least } 27 \text{ of indexes } i\text{'s with } 1 \leq i \leq 535 \\
		&\hspace{-7em} \text{such that } \hat{\epsilon}_i \sim \Omega_i \text{ with } \Omega_i^{-1}(0.95) < -1.6551, \text{ or } \Omega_i^{-1}(0.05) > 1.5533 \bigr\},
	\end{split}
\end{equation*}
where $\Omega_i^{-1}$ represents the inverse uncertainty distribution of $\epsilon_i$. 
Figure \ref{fig:application3} plots $\Omega_i^{-1}(0.05)$ and $\Omega_i^{-1}(0.95)$ for $i = 1,2,\dots,535$, along with the critical boundaries of the test. It can be observed that no sample point falls into the rejection region $W_3$. Consequently, the null hypothesis $H_0$ cannot be rejected. This result demonstrates that the established model fits the sample data well.\\
\indent According to formula $(\ref{eq:13})$, given a new set of explanatory variables $(x_1, x_2, x_3, x_4, x_5)=(0.5, 0, 1, -0.5, 0.2)$, we derive the expected value of the new response variable $\hat{y}$ as follows:
\begin{equation*}
	\hat{\mu}=\hat{g}(\hat{\beta}_1x_{1}+\hat{\beta}_2x_{2}+\hat{\beta}_3x_{3}+\hat{\beta}_4x_{4}+\hat{\beta}_5x_{5})+\hat{e}=28.1022.
\end{equation*}
Given the significance level $\alpha=0.1$, the minimum value satisfying the inequality ($\ref{eq:14}$) is obtained as $a^{*}=0.1387$. Accordingly, the confidence interval for the response variable $\hat{y}$ is calculated as $[\mu-a^{*}, \mu+a^{*}]=[27.9635, 28.2409]$.

\begin{figure}[H]
	\centering
	\begin{subfigure}{\linewidth}
		\centering
		\setlength{\parskip}{0pt}
		\setlength{\baselineskip}{0pt}
		\caption{}  
		\includegraphics[width=0.8\textwidth,height=6.0cm]{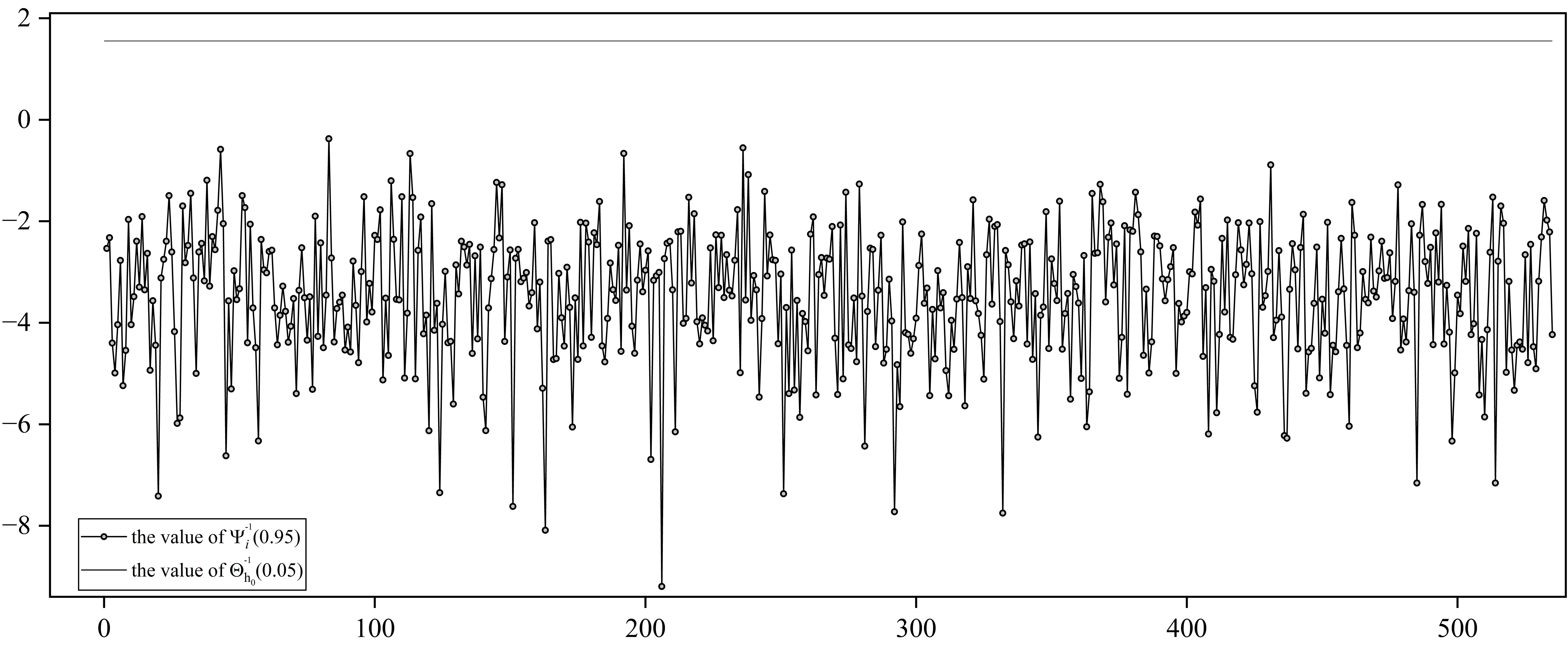} 
		\label{fig:application3.1}
	\end{subfigure}
	\begin{subfigure}{\linewidth}
		\centering
		\caption{}  
		\includegraphics[width=0.8\textwidth,height=6.0cm]{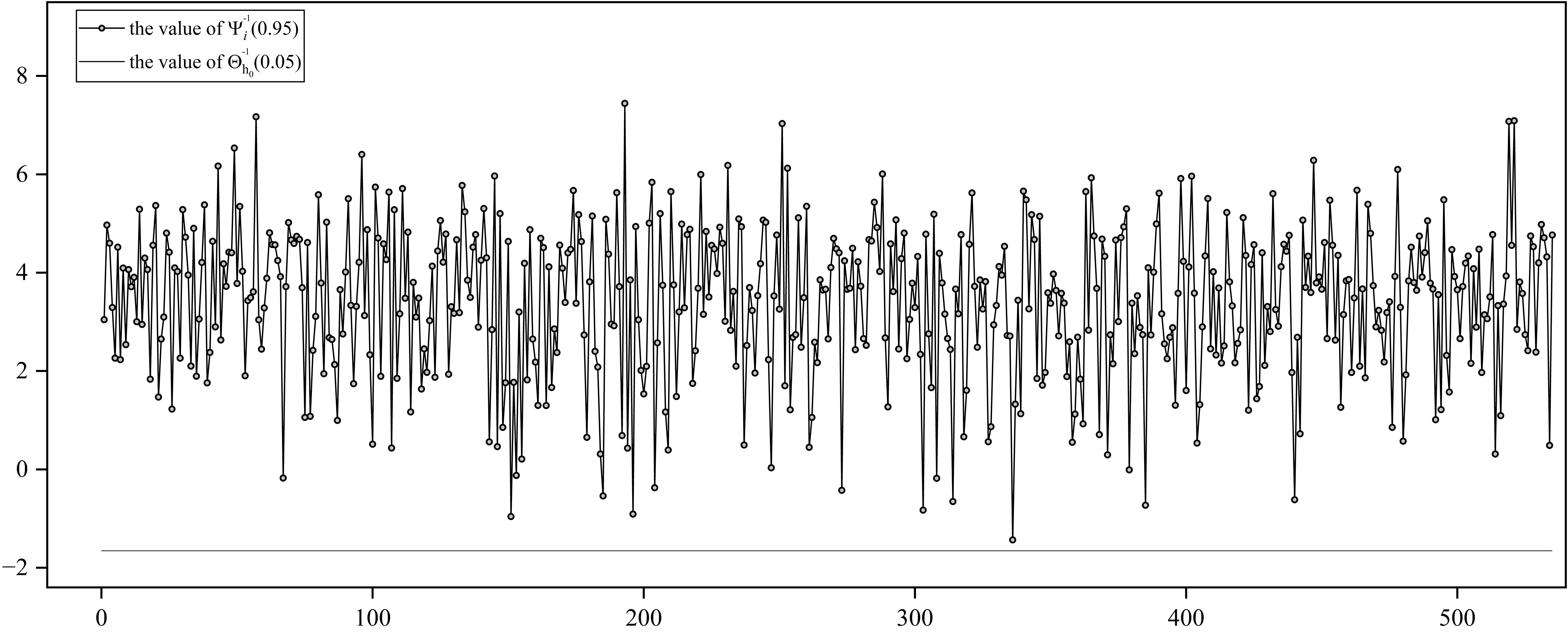} 
		\label{fig:application3.2}
	\end{subfigure}
	
	\caption{Scatter plots of $\Omega_i^{-1}(0.05)$ and $\Omega_i^{-1}(0.95)$ for $i = 1, \dots, 500$ with the boundaries of the test $W_3$.}
	\label{fig:application3}
\end{figure}

\subsection{Significance test}
\indent Since the coefficient $\beta_2=0.0381$ is close to zero, we conduct an uncertain significance test by Ye and Liu (\citeyear{YeLiu2023}) to examine whether $\beta_2$ is significant in the model, as presented below:\\
\begin{equation*}
H_0: \beta_2 = 0  \quad \text{v.s.} \quad H_1: \beta_2 \neq 0.
\end{equation*}
\indent In the significance test, we need to perform two hypothesis tests under null hypothesis($H_0$) and alternative hypothesis($H_1$). Following the decision criterion, $\beta_2$ is considered insignificant when $(\epsilon_{0,1}, \epsilon_{0,2},\ldots, \epsilon_{0,n})$ $\notin$ $W_{(0)}$ and $(\epsilon_{1,1}, \epsilon_{1,2},\ldots, \epsilon_{1,n})$ $\in$ $W_{(1)}$. Otherwise, $\beta_2$ is significant. The complete testing procedure is illustrated in Steps 1 to 7.\\
\textbf{Step 1}: ($H_0$: Parameter estimation) Under the null hypothesis $H_0: \beta_2=0$, we obtain the estimators of coefficients via formula (\ref{eq:7}) as follows:
\begin{equation*}
	\bm{\hat{\beta}}=(\beta_1,\beta_2,\beta_3,\beta_4,\beta_5)^T=(0.2084, 0, 0.8695, 0.3912, 0.2180)^T
\end{equation*}
\textbf{Step 2}: ($H_0$: Residual analysis) From equation (\ref{eq:16}), the residual sequence $(\varepsilon_{0,1}, \ldots, \varepsilon_{0,535})$ is obtained. Then, we calculate the expected value and variace of the residual sequence as follows:
\begin{equation*}
	\hat{e}_0=\frac{1}{535} \sum_{i=1}^{535}\hat{\varepsilon}_{0,i}=0.49, \quad \hat{\sigma}^2_0=\frac{1}{535} \sum_{i=1}^{535} (\hat{\varepsilon}_{0,i} - \hat{e}_0)^2=1.04^2.
\end{equation*}
\textbf{Step 3}: ($H_0$: Hypothesis test) Based on the results above, we assume that disturbance term follows a normal uncertain distribution $\mathcal{N}(0.49,1.04)$. At the significance level of $\alpha= 0.05$, we have $\Theta_{h_0}^{-1}(0.05) = -1.1934$, $\Theta_{h_0}^{-1}(0.95) = 2.1832$. Then, the rejection region can be expressed as:
\begin{equation*}
    \begin{split}
        W_{(0)} = \bigl\{ \left( \hat{\epsilon}_{0,1}, \hat{\epsilon}_{0,2}, \dots, \hat{\epsilon}_{0,535} \right) \mid 
        &\text{there are at least } 27 \text{ of indexes } i\text{'s with } 1 \leq i \leq 535 \\
        &\hspace{-8.5em} \text{such that } \hat{\epsilon}_{0,i} \sim \Omega_{0, i} \text{ with } \Omega_{0, i}^{-1}(0.95) < -1.1934, \text{ or } \Omega_{0, i}^{-1}(0.05) > 2.1832 \bigr\},
    \end{split}
\end{equation*}
where $\Omega_{0, i}^{-1}$ represents the inverse uncertainty distribution of $\varepsilon_{0,i}$. As shown in figure \ref{fig:significance}, no residual falls in rejection region. Thus, we have $(\hat{\epsilon}_{0,1}, \hat{\epsilon}_{0,2}, \dots, \hat{\epsilon}_{0,535})$ $\notin$ $W_{(0)}$.\\
\textbf{Step 4-6}: (The test for $H_1$) The procedure under the alternative hypothesis $H_1: \beta_2 \neq 0$ follows Section 7.2 exactly and is omitted here. As a direct result,  $(\hat{\epsilon}_{0,1}, \hat{\epsilon}_{0,2}, \dots, \hat{\epsilon}_{0,535})$ $\notin$ $W_{(1)}$.\\
\textbf{Step 7}: (Determine) Based on the decision criterion, we determine that $\beta_2$ is significant.

\subsection{Model comparison}
In this subsection, we perform a comparative analysis to demonstrate the superiority of the single-index model with unspecified function $g$ in data fitting. For consistency, the least squares method is used to fit all models to the real-world data in this section. Based on the distribution characteristics of the real data, four parametric forms of $g$ are considered: identity, quadratic, logarithmic and exponential functions. The different model forms and fit results are presented in Table \ref{tab:model_comparison}.\\

\begin{table}[h]
	\centering
	\caption{Comparison of different link function models}
	\label{tab:model_comparison}
	\begin{tabular}{ccccccc}  
		\toprule
		$g$ & $\beta_1$ & $\beta_2$ & $\beta_3$ & $\beta_4$ & $\beta_5$ & Variance \\
		\midrule   
		$z$&-0.4114&-0.0516&-0.2537&-0.7907&-0.2535 & 1.6030 \\
		$z^2$& -0.0156&-0.0018&-0.0081&-0.0276&-0.0087& 1.6251\\
		$Aexp(z)$& -0.0156&-0.0018&-0.0081&-0.0276&-0.0087& 1.2749\\
		$ln(1+z^2)$& $-1.1\times10^6$ & $-1.4\times10^5$ & $-7.2\times10^5$ & $-2.2\times10^6$ & $-7.1\times10^5$ & 8.6860 \\
		Unknown  & 0.1770&0.0381&0.8952&0.3438&0.2183& $\bm{0.9765}$\\
		\bottomrule
	\end{tabular}
\end{table}

\indent According to the residual variance criterion, lower values indicate a better model fit. As shown in Table \ref{tab:model_comparison}, $\sigma_{\text{single}}^2$ is minimized among the models with different specified link functions. Thus, the uncertain single-index model provides a better fit to the data, demonstrating greater flexibility in capturing complex relationships between variables.

\begin{figure}[H]
	\centering
	\begin{subfigure}{\linewidth}
		\centering
		\setlength{\parskip}{0pt}
		\setlength{\baselineskip}{0pt}
		\caption{}  
		\includegraphics[width=0.8\textwidth,height=6.0cm]{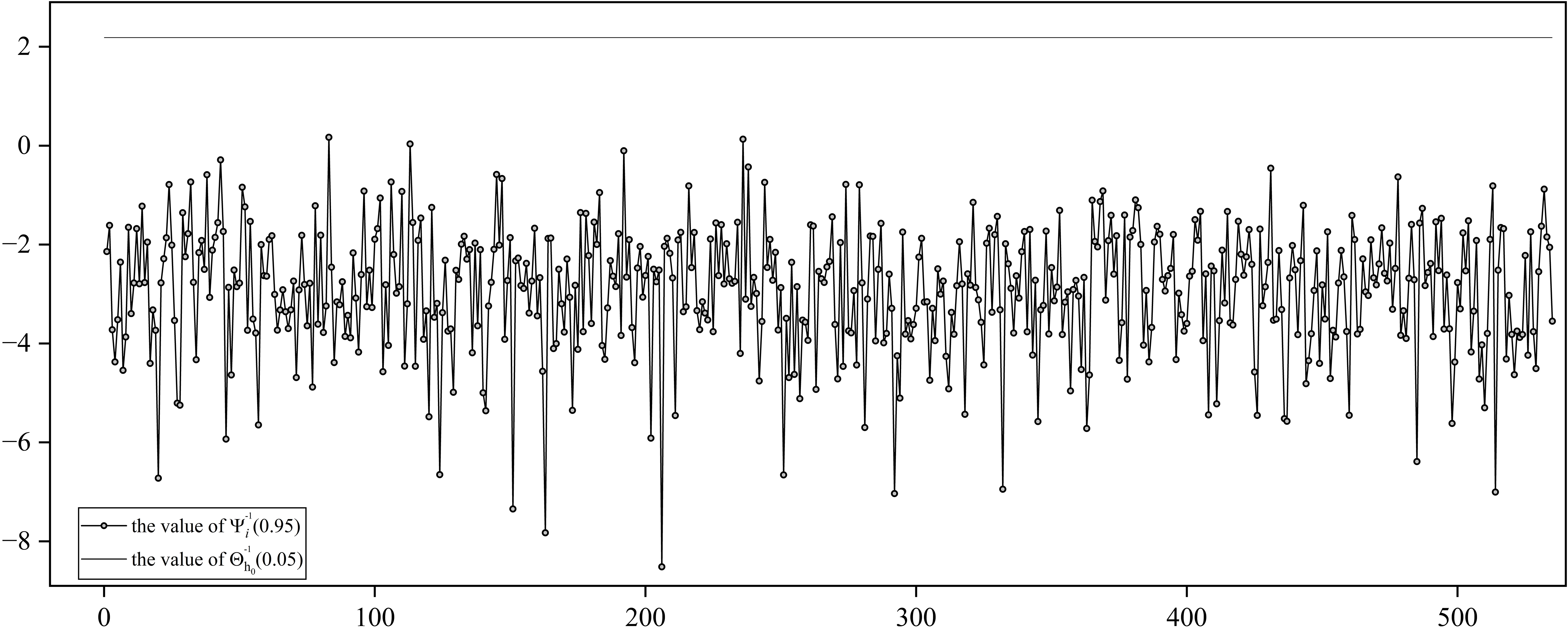} 
		\label{fig:significance1}
	\end{subfigure}
	\begin{subfigure}{\linewidth}
		\centering
		\caption{}  
		\includegraphics[width=0.8\textwidth,height=6.0cm]{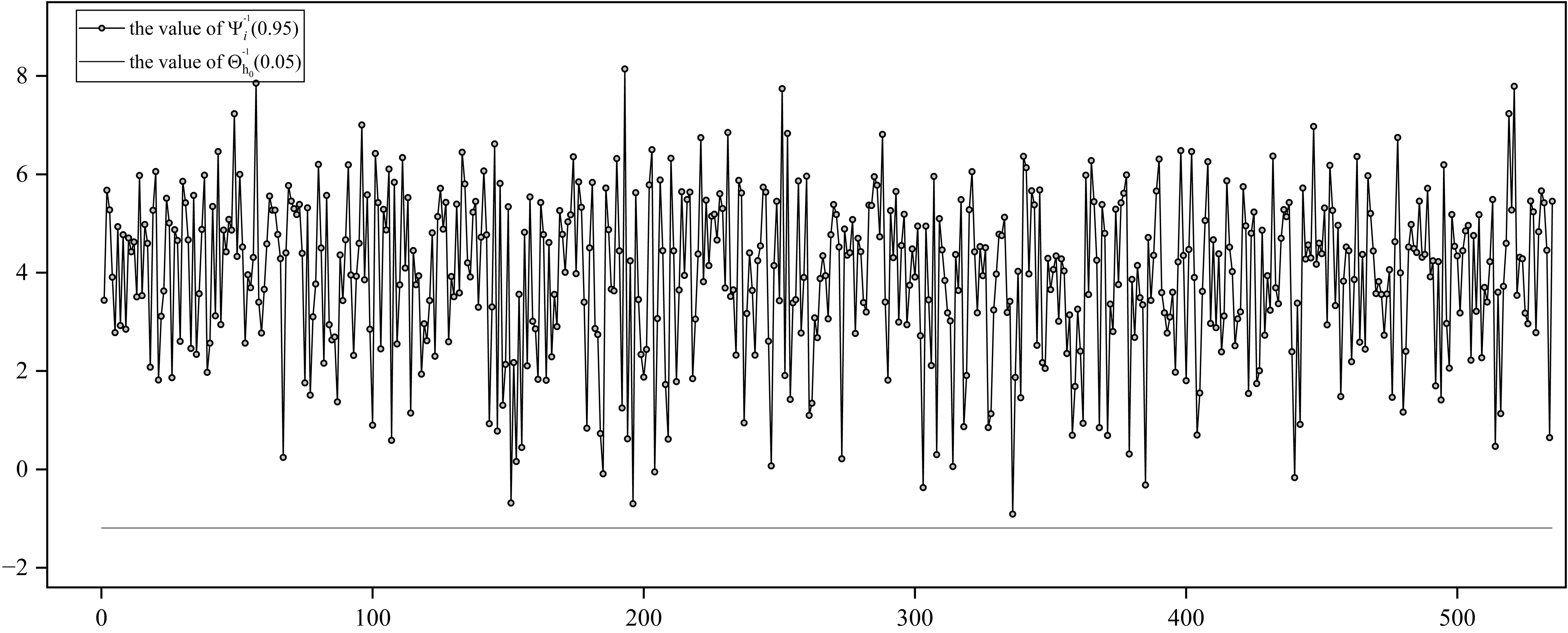} 
		\label{fig:significance2}
	\end{subfigure}
	
	\caption{Scatter plots of $\Omega_i^{-1}(0.05)$ and $\Omega_i^{-1}(0.95)$ for $i = 1, \dots, 500$ with the boundaries of the test $W_{(0)}$.}
	\label{fig:significance}
\end{figure}

\section{Conclusion}\label{sec7}
This paper proposes two types of uncertain single-index models (USIC and USIU) for imprecise data in uncertain environments to meet modeling requirements in high-dimensional predictor and complex variable-relationship scenarios. For the USIC model, an uncertain Nadaraya–Watson kernel estimator is introduced to approximate the unknown link function, and a two-step iterative estimation procedure is employed to estimate both the index coefficients and the unknown function. In the USIU model estimation, under the monotonicity constraint, the B-spline method is used to approximate the unknown function. This is because the N-W kernel estimator lacks an explicit expression; the B-spline method offers greater convenience for differentiation and facilitates the imposition of monotonicity. The parameters of the USIU model are also estimated using the same two-step iterative framework. To evaluate the effectiveness of the proposed estimation methods, two simulation studies are conducted, followed by residual analysis of the fitted models to assess their fitting performance. Simulation results demonstrate that both models can effectively capture the complex functional relationships inherent in imprecise data. Furthermore, a real-data analysis based on weather data is performed. To verify whether coefficients close to zero are truly insignificant, an uncertain significance test is applied. Finally, we compare the fitting performance with four uncertain regression models and validate the practical applicability of the proposed approach.

Future research should focus on statistical inference for uncertain partially linear single-index models and uncertain varying-coefficient single-index models, to enhance their performance in practical applications. Meanwhile, more efforts should be made to extend the estimation theory for uncertain single-index models and further compare the applicability of different estimation methods across various scenarios.

\section*{Funding}
The work was supported by the National Natural Science Foundation of China (12561047), the Xinjiang Talent Development Fund (XJRC-2025-KJ-PY-KJLJ-108), and the 2025 Central Guidance for Local Science and Technology Development Fund (ZYYD2025ZY20).

\end{document}